\documentclass[final,1p,times,number,amsthm]{elsart}
\journal{Journal of Symbolic Computation}

\usepackage[utf8x]{inputenc}
\usepackage{yjsco}
\usepackage{natbib}
\usepackage{amsmath}
\usepackage{amsfonts}
\usepackage{amssymb}
\usepackage{amsthm}
\usepackage{fontenc}
\usepackage{graphicx}
\usepackage{bbm}
\usepackage{tikz}
\usepackage{todonotes}
\usepackage{multirow}
\usepackage[colorlinks]{hyperref}

\newtheorem{lemma}{Lemma}
\newtheorem{theorem}[lemma]{Theorem}
\newtheorem{corollary}[lemma]{Corollary}
\newtheorem{example}[lemma]{Example}
\newtheorem{definition}[lemma]{Definition}
\newtheorem{proposition}[lemma]{Proposition}
\newtheorem{remark}[lemma]{Remark}

\usetikzlibrary{patterns,decorations.pathreplacing}

\newlength\thickness
\pgfdeclarepatternformonly[\thickness]{section}
{\pgfpointorigin}{\pgfpoint{0.2cm}{0.2cm}}{\pgfpoint{0.2cm}{0.2cm}}
{
  \pgfsetlinewidth{\thickness}
  \pgfpathmoveto{\pgfpoint{0cm}{0cm}}
  \pgfpathlineto{\pgfpoint{1cm}{1cm}}
  \pgfpathclose
  \pgfsetlinewidth{\thickness}
  \pgfpathmoveto{\pgfpoint{0cm}{.5cm}}
  \pgfpathlineto{\pgfpoint{.5cm}{1cm}}
  \pgfpathclose
  \pgfsetlinewidth{\thickness}
  \pgfpathmoveto{\pgfpoint{.5cm}{0cm}}
  \pgfpathlineto{\pgfpoint{1cm}{.5cm}}
  \pgfusepath{stroke}
}

\pgfdeclarepatternformonly[\thickness]{section2}
{\pgfpointorigin}{\pgfpoint{1cm}{1cm}}{\pgfpoint{1cm}{0.2cm}}
{
  \pgfsetlinewidth{\thickness}
  \pgfpathmoveto{\pgfpoint{0cm}{1cm}}
  \pgfpathlineto{\pgfpoint{1.0cm}{0.cm}}
  \pgfpathclose
  \pgfsetlinewidth{\thickness}
  \pgfpathmoveto{\pgfpoint{0cm}{.5cm}}
  \pgfpathlineto{\pgfpoint{.5cm}{1cm}}
  \pgfpathclose
  \pgfsetlinewidth{\thickness}
  \pgfpathmoveto{\pgfpoint{.5cm}{0cm}}
  \pgfpathlineto{\pgfpoint{1cm}{.5cm}}
  \pgfusepath{stroke}
}

\pgfdeclarepatternformonly[\thickness]{section3}
{\pgfpointorigin}{\pgfpoint{0.3cm}{0.3cm}}{\pgfpoint{0.15cm}{0.15cm}}
{
  \pgfsetlinewidth{0.3pt}
  \pgfpathmoveto{\pgfqpoint{0pt}{0pt}}
  \pgfpathlineto{\pgfqpoint{3pt}{3pt}}
  \pgfpathmoveto{\pgfqpoint{0pt}{3pt}}
  \pgfpathlineto{\pgfqpoint{3pt}{0pt}}
  \pgfusepath{stroke}
}

\pgfdeclarepatternformonly[\thickness]{section4}{\pgfqpoint{-1pt}{-1pt}}
{\pgfqpoint{4pt}{4pt}}{\pgfqpoint{\thickness}{\thickness}}
{
  \pgfsetlinewidth{0.3pt}
  \pgfpathmoveto{\pgfqpoint{0pt}{0pt}}
  \pgfpathlineto{\pgfqpoint{0pt}{3.1pt}}
  \pgfpathmoveto{\pgfqpoint{0pt}{0pt}}
  \pgfpathlineto{\pgfqpoint{3.1pt}{0pt}}
  \pgfusepath{stroke}
}

\tikzset{
  thickness/.store in = \thickness,
  thickness = 0.5pt
}

\let\set\mathbbm
\newcommand{\cfield}{\mathbbm{C}}

\newcommand{\ps}[2]{\cfield[[x_{#1},\dots,x_{#2}]]}

\newcommand{\pder}[2]{\frac{\partial}{\partial #2}#1}
\newcommand{\xqedhere}[2]{%
  \rlap{\hbox to#1{\hfil\llap{\ensuremath{#2}}}}}

\begin{document}

\begin{frontmatter}

\date{\today} \title{Formal Solutions of Completely Integrable Pfaffian Systems
  With Normal Crossings}

\author{Barkatou, Moulay A.}
\address{XLIM UMR 7252 , DMI, University of Limoges; CNRS\\
        123, Avenue Albert Thomas, 87060 Limoges, France\\
        \texttt{moulay.barkatou@unilim.fr}}

 \author{Jaroschek, Maximilian}
 \address{Max Planck Institute for Informatics,\\ Saarbruecken, Germany\\
  \texttt{maximilian@mjaroschek.com}}

\author{Maddah, Suzy S.}
\address{Fields Institute\\
  222 College St, Toronto, ON M5T 3J1 Ontario, Canada\\
  \texttt{suzy.maddah@inria.fr}}

\begin{abstract}
  In this paper, we present an algorithm for computing a fundamental matrix of
  formal solutions of completely integrable Pfaffian systems with normal
  crossings in several variables. This algorithm is a generalization of a method
  developed for the bivariate case based on a combination of several reduction
  techniques and is partially\footnote{Our Maple package PfaffInt can be
    downloaded at: \url{http://www.mjaroschek.com/pfaffian/} It contains
    functionalities illustrated by examples for the splitting, column reduction,
    rank reduction, and the computation of exponential parts of multivariate
    completely integrable systems with normal crossings}implemented in the
  computer algebra system \textsc{Maple}.
\end{abstract}


\end{frontmatter}

%%%%%%%%%%%%%%%%%%%%%%%%
%%        BODY        %%
%%%%%%%%%%%%%%%%%%%%%%%%

\section{Introduction}
\label{sec:intro}
Pfaffian systems arise in many applications~\cite{key74}, including the studies
of aerospace, celestial mechanics~\cite{key31}, and
statistics~\cite{key3062}. So far, the most important systems for applications
are those with so-called normal crossings~\cite{key32}.

A univariate completely integrable Pfaffian system with normal crossings reduces
to a singular linear system of ordinary differential equations (ODS, in short),
which have been studied extensively (see~\cite{key6,key7} and references
therein). Moreover, unlike the general case of several variables considered
herein, algorithms to related problems leading to the computation of formal
solutions have been developed by various authors
(see~\cite{key24,key40,key25,key26} and references therein). The \textsc{Maple}
package \textsc{Isolde}~\cite{key27} and \textsc{Mathemagix} package
\textsc{Lindalg}~\cite{key427} are dedicated to the symbolic resolution of such
systems.

More recently, bivariate systems were treated by the first and third author of
this paper in~\cite{key101}. This paper refines the results of the bivariate
case and generalizes them to treat the more general multivariate case.

To get an intuition of the kind of systems we consider, we informally study the
following simple bivariate completely integrable Pfaffian system with normal
crossings. A formal definition of these systems will be given in
Section~\ref{sec:prelim}.

\begin{example}~\cite[Example 2]{key101}
  \label{ex:sim}
  Given the following bivariate system over the ring of formal power series in
  $(x_1,x_2)$ with complex coefficients:
  \begin{equation*}
    \begin{cases}
      x_1^{4} \pder{F}{x_1} = A_1 F = \left(\begin{matrix}
          x_1^3 + x_1^2+x_2 & x_2^2 \\
          -1 & x_1^3 + x_1^2 -x_2
        \end{matrix}\right) F, \\
      x_2^3 \pder{F}{x_2} = A_2 F= \left(\begin{matrix}
          x_2^2 -2 x_2 -6 & x_2^3 \\
          -2 x_2 & -3 x_2^2 -2 x_2 -6
        \end{matrix}\right) F,
    \end{cases}
  \end{equation*}
  we are interested in constructing the formal objects $F$ that satisfy the
  system. The existence of a fundamental matrix of solutions and its general
  form follows from well known theoretical results (see
  Corollary~\ref{gerardsol}). The proof, however, is not constructive. For
  simplicity, we assume we already know that a fundamental matrix of formal
  solutions in our particular case is of the form
  \begin{equation} \label{sim:sol1} \Phi(x_1, x_2) x_1^{C_1}x_2^{C_2}
    e^{q_1(x_1^{- 1/s_1})} e^{q_2(x_2^{- 1/s_2})}, \end{equation} where
  $\Phi(x_1, x_2)$ is a matrix with formal power series entries, $C_1$ and $C_2$
  are matrices with entries in $\set C$, and $q_1$,~$q_2$ are polynomials in
  $\set C[z_1], \set C[z_2]$ respectively. We now want to determine
  $ \Phi, C_1, C_2, q_1, q_2,s_1,s_2$.  For this purpose, we use the algorithm
  presented in~\cite{key101}. The main idea is to compute one part
  of the solution by considering an associated ODS in only one variable and then
  use this information to compute the other parts of the solution by
  transforming and decoupling the system into smaller and simpler systems:

  \begin{itemize}
  \item First, we construct two associated systems whose equations are derived
    by setting either $x_1=0$ or $x_2=0$:
    \[
    \begin{cases}
      x_1^{4} \pder{F}{x_1} = A_1 (x_1,0) \; F= \left(\begin{matrix}
          x_1^3+ x_1^2 & 0 \\
          -1 & x_1^3 + x_1^2
        \end{matrix}\right) F,\\
      x_2^3 \pder{F}{x_2} = A_2 (0, x_2) \; F= \left(\begin{matrix}
          x_2^2 -2 x_2 -6 & x_2^3 \\
          -2 x_2 & -3 x_2^2 -2 x_2 -6
        \end{matrix}\right)F.
    \end{cases}
    \]
    We show in Section~\ref{sec:invariants} that the formal invariants $q_1,
    q_2,s_1$ and $s_2$ can be computed from these associated systems. Via
    \textsc{Isolde} or \textsc{Lindalg} we compute $s_1=s_2=1$ and
    $q_1(1/x_1)=\frac{-1}{x_1}$, $q_2(1/x_2)=(\frac{3}{x_2^2} + \frac{2}{x_2} ),$
    and~\eqref{sim:sol1} becomes
    \[
    \Phi (x_1, x_2) x_1^{{C}_1} x_2^{C_2} e^{\frac{-1}{x_1}} e^{\frac{3}{{x_2}^2}
      + \frac{2}{x_2}}.
    \]
  \item Next, we apply the so-called eigenvalue shifting
    $F = e^{\frac{-1}{x_1}} e^{\frac{3}{x_2^2} + \frac{2}{x_2}} G$ (for a new
    unknown vector $G$), to facilitate the next step. The shifting
    yields:\goodbreak
    \[
    \begin{cases}
      x_1^{4} \pder{G}{x_1} =  \left(\begin{matrix}
          x_1^3+x_2 & x_2^2 \\
          -1 & x_1^3 -x_2
        \end{matrix}\right) G,  \\
      x_2^2 \pder{G}{x_2} =  \left(\begin{matrix}
          x_2  & x_2^2 \\
          -2 & -3 x_2
        \end{matrix}\right) G.
    \end{cases}
    \]
  \item After the eigenvalue-shifting we apply another transformation that
    reduces the orders of the singularities in $x_1$ and $x_2$ to their minimal
    integer values. By setting $G = T_1 H$ where
    \[
    T_1=\left(\begin{matrix}
        x_2 x_1^3 & -x_2 \\
        0 & 1
      \end{matrix}\right),
    \]
    we get:
    \begin{equation*}
      \begin{cases}
        x_1 \pder{H}{x_1} = \left(\begin{matrix}
            -2 & 0 \\
            -x_2 & 1
          \end{matrix}\right) H,  \\
        x_2 \pder{H}{x_2} = \left(\begin{matrix}
            -2 & 0\\
            -2 x_1^3 & -1
          \end{matrix}\right) H.
      \end{cases}
    \end{equation*}
  \item Finally, via some linear algebra (see~\cite[Chapter 3]{key73} for
    general cases) we compute the transformation
    \[
    T_2=\left(\begin{matrix}
        1 & 0 \\
        \frac{x_2}{3} + 2 x_1^3 &\hspace{0.2cm}-1
      \end{matrix}\right),
    \]
    and setting $H=T_2 U$ results in the system
\[
\begin{cases}
  x_1 \pder{U}{x_1} = C_1 U = \left(\begin{matrix} -2 & 0 \\ 0 &
      1 \end{matrix}\right) U ,\\
  x_2 \pder{U}{x_2} = C_2 U = \left(\begin{matrix} -2 & 0 \\ 0 &
      -1 \end{matrix}\right)U.
\end{cases}
\]

We can now read off $C_1$ and $C_2$. We collect the applied
transformations and get a fundamental matrix of solutions:
    \[
      \underbrace{T_1 T_2}_{=:\Phi} x_1^{C_1} x_2^{C_2}e^{\frac{-1}{x_1}}
      e^{\frac{3}{x_2^2} + \frac{2}{x_2}},
    \]
    where $C_1=\left(\begin{matrix} -2 & 0 \\ 0 & 1 \end{matrix}\right)$ and
    $C_2=\left(\begin{matrix} -2 & 0 \\ 0 & -1 \end{matrix}\right)$.
  \end{itemize}
\end{example}

Unlike this simple example, the steps of computation can be far more involved
and demand multiple levels of recursion. In order to generalize this algorithm
to more than two variables, the following nontrivial questions have to be
addressed:

\begin{itemize}
\item Can the information on the formal invariants still be obtained from the
  associated ODS systems?
\item Can a rank reduction algorithm be developed without relying on properties
  of principal ideal domains as in the univariate and bivariate case?
\end{itemize}
The results of~\cite{key101} which have an immediate generalization to the
multivariate setting are refined herein and supported by fully transparent
proofs and illustrating figures (Theorem~\ref{exponentialpfaff} which answers
positively and unconditionally the first question, the structure of the main
algorithm described in Section~\ref{sec:outline}, and
Theorem~\ref{moserpfaff}). However, the answer to the second question is more
elaborate (see Section~\ref{sec:rankred}) and requires the discussion of two
problems which are not discussed in~\cite{key101}:
\begin{itemize}
\item The major obstacle to a generalization of the results on the
  bivariate case lies in the process of finding integral relations among
  generators of certain modules over the ring of multivariate power series. In
  Section~\ref{colpfaff}, we propose a solution that relaxes the condition
  of working over a principal ideal domain to working over local rings and show
  how to utilize Nakayama’s Lemma in the formal reduction process if the modules
  under investigation are free.
\item We discuss an algorithmic difficulty which arises as not all formal power
  series under manipulation admit a finite representation, even if the input
  Pfaffian system is given in a finite form (see
  Example~\ref{exm:trunc}). Although this problem arises in the bivariate case
  as well, it has not been addressed before (neither in~\cite{key101} nor
  in~\cite{key5,key73}). We provide a reasoning to check the correctness of our
  algorithm.
\end{itemize} 
We thus present the first comprehensive description of the state of the art
algorithmic approach for solving completely integrable Pfaffian systems with
normal crossings in the multivariate setting. Our investigation also involves
the multivariate versions of the transformations used classically in the
well-studied univariate case (e.g. shearing transformations in
Section~\ref{shearpfaff}, column reduction transformations in
Section~\ref{colpfaff}, and properties of transformations in
Proposition~\ref{Tform}). Not only does this discussion serve the manipulation
of such transformations within our proposed formal reduction, but it also plays
a role in future generalizations of many other algorithms available for
univariate systems (e.g. the alternative rank reduction algorithm of
Section~\ref{sec:alt} and the notion of \emph{simple systems} as suggested in
the conclusion).

This paper is divided as follows: In Section~\ref{sec:prelim}, we recall the
basic definitions and the necessary theory for our algorithm. This includes the
general form of the solutions, the notion of equivalence between systems, the
classification of singularities, and a description of the necessary
transformations whose generalization to the multivariate case is
straightforward. In Section~\ref{sec:outline}, we give the general structure of
our proposed algorithm which relies on two major components: The first is
associating to our system a set of ODS's from which its formal invariants can be
efficiently derived. This is detailed in Section~\ref{sec:invariants}. The
second component is the rank reduction which we give in
Section~\ref{sec:rankred}. The main algorithm is then given in
Section~\ref{mainpfaff} before concluding in Section~\ref{conpfaff}.

\section{Preliminaries}
\label{sec:prelim}
\subsection{Completely Integrable Pfaffian Systems with Normal Crossings}
\label{sec:prelimpfaff}
The systems considered in this paper are those whose associated differential
form is a 1-form. More explicitly, let ${\rm R}:=\ps{1}{n}$ be the ring of
formal power series in $x_1, x_2,\dots, x_n$ over the field of complex numbers
$\set C$. A Pfaffian system with normal crossings is a system of linear partial
differential equations of the form
\begin{equation}
  \label{eq:sys}
x_i^{p_i+1} \pder{F}{x_i} = A_iF, \quad 1 \leq  i \leq n,
\end{equation}
where the $A_i$'s are $d \times d$ matrices with entries in $\rm R$.  The
system~\eqref{eq:sys} is completely determined by the $A_i$'s and $p_i$'s and we
conveniently denote it by $[A]$. Each of the $p_i$'s is an integer and the
number $p’_i:=\max(0,p_i)$ is called the Poincar\'e rank of the $i^{th}$
component $A_i$.  The $n$-tuple
$$ p := (p'_1, \dots, p'_n)$$ is called the Poincar\'e rank of the system $[A]$.
% Each of the $p_i$'s is an integer and is called the
% Poincar\'e rank of the $i^{th}$ component $A_i$.  The $n$-tuple
% $$ p := (p_1, \dots, p_n)$$ is called the Poincar\'e rank of the system
% $[A]$. 
% Without loss of generality, the singularity is placed at the
% origin. Otherwise, translations in the independent variables can be
% performed. 
If $p_i \leq -1$ for every $i \in \{1, \dots, n\}$ then the origin is an
ordinary (non-singular) point of the system and the system is said to be
regular. In this paper, we tackle the rather more interesting singular
systems. The singular locus of a system with normal crossings is a union of
hyperplanes of coordinates $x_1 x_2 \dots x_n =0$. A Pfaffian system is called
completely integrable, if the following commutation rule holds for all
$i,j\in\{1,\dots,n\}$:
\begin{equation}
  \label{eq:cond}
  A_iA_j - A_jA_i = x_i^{p_i+1}\pder{A_j}{x_i} - x_j^{p_j+1}\pder{A_i}{x_j}.
\end{equation}
Subsequently, whenever we refer to a Pfaffian system, we assume it is a
completely integrable system with normal crossings. For the remainder of this
paper we once and for all fix a Pfaffian system $[A]$ for which all the $A_i$
are non-zero and there is at least one strictly positive $p_i$. All subsequent
definitions and theorems are stated in this setting, disregarding systems for
which the origin is an ordinary point.

\subsection{Notations and Algebraic Structures}
Our notations follow a set of guidelines in order to help the reader remember
the multitude of different objects involved in our work. Single letter
identifiers are usually chosen to be the initial letter of the mathematical term
attached to the referenced object, like $d$ for dimension and~$\rm R$ for a
ring. For a vector $v$, its $i^{th}$ component is given by $v_i$ and for a
univariate power series $s$ the $i^{th}$ coefficient is denoted by $s_i$. We do
not distinguish between row and column vectors. Upper case letters are used for
algebraic structures, matrices and the unknown in a Pfaffian system. A family of
matrices is given with lower indices, e.g.\ $(M_{i,j})_{i,j\geq 0}$, and for a
matrix $M_{i,j}$, blocks are given with upper indices, e.g.
\[
M_{i,j}=\left(\begin{matrix}
    M_{i,j}^{11} & M_{i,j}^{12}\\
    M_{i,j}^{21} & M_{i,j}^{22}
  \end{matrix}\right),
\]
where the size of the different blocks are clear from the context.
By $x$ we denote the collection of variables $x_1,\dots,x_n$ and we use
$\bar{x}_i$ to refer to the variables
$x_1,\dots,x_{i-1},x_{i+1},\dots,x_n$.

One can expand the $A_i$ in system~$[A]$ as a formal power series with respect
to $x_i$:
\begin{equation*}
  x_i^{p_i +1 }\pder{F}{x_i} = (A_{i,0}+ A_{i,1} x_i +
  A_{i,2} x_i^2 + \dots ) F,
\end{equation*}
where the $A_{i,j}$ are elements of $\set C[[\bar{x}_i]]$. We denote this ring
by ${\rm R}_{\bar{x}_i}$.  The first coefficient $A_{i,0} = A(x_i=0)$ in such an
expansion can be regarded as non-zero without any loss of generality, otherwise
$p_i$ can be readjusted. We call $A_{i,0}$ the leading matrix coefficient of the
$i^{th}$ component.

Aside from the rings $\rm R$ and ${\rm R}_{\bar{x}_i}$, we will frequently have
to work in other algebraic structures. We denote by
${\rm K}:=\operatorname{Frac}({\rm R})$ (respectively ${\rm K}_{\bar{x}_i}$) the
fraction field of $\rm R$ (respectively ${\rm R}_{\bar{x}_i}$).  Let $L$ be the set of
monomials given by,
$$ L = \{ x^{\beta} = x_1^{\beta_1} x_2^{\beta_2} \dots x_n^{\beta_n}\text{ for
} \beta = (\beta_1, \beta_2, \dots, \beta_n) \in \mathbb{N}^n \}. $$
Clearly, $L$ is closed under multiplication and contains the unit element. Then,
one can define ${\rm R}_{L} := L^{-1} \rm R$, the localization of $\rm R$ at
$L$, i.e.\ the ring of series with only finitely many terms having monomials of
strictly negative exponents. Unlike in the univariate case, $\rm K$ and
${\rm R}_L$ do not refer to the same algebraic structure, e.g.\ $(x_1+x_2)^{-1}$
is an element of $\rm K$ which is not an element of ${\rm R}_L$. In fact, there
exists no $\beta \in \mathbb{N}^2$ such that
$x^{\beta} (x_1+x_2)^{-1} \in \rm R$:
$(x_1+x_2)^{-1} = \frac{1/x_1}{1-(-x_2/x_1)} $ whose formal expansion with
respect to $x_2$ is $\sum_{i\geq 0}(-1)^ix_1^{-i-1}x_2^i$, which has infinitely
many poles in $x_1$. For a further characterization of $\rm K$, one may refer
to~\cite{key3908}.

Finally, in the sequel it will be necessary to introduce ramifications of the
form $x_i=t_i^{\alpha_i}$ for new variables $t_i$ and positive integers
$\alpha_i$. We will therefore write ${\rm R}_t$ for $\set C[[t_1,\dots,t_n]]$
and allow analogous notations for all structures introduced so far. The identity
and zero matrices of dimension $d$ are denoted by $I_d$ and $O_d$, and we set: $\rm{GL}_d (R) = \{ M \in\operatorname{Mat}_{d\times d}({\rm R})\; | \; \operatorname{det} (M) \; \text{is invertible in} \; \rm R \}.$
 
\subsection{Equivalent Systems}
As we have seen in the introductory example, we will make use of transformations
to bring a system into particular forms. Such a transformation acts on a
Pfaffian system as follows: A linear transformation (also called gauge
transformation) $F=T G$, where
$T \in GL_d({\rm K})$, applied to~\eqref{eq:sys} results in the system 
\begin{equation}
  \label{eq:equiv}
x_i^{\tilde{p}_i+1} \pder{G}{x_i} = \tilde{A}_i G, \quad 1 \leq i \leq n,
\end{equation}
where
\begin{equation}
  \label{eq:gauge}
  \frac{\tilde{A}_i }{x^{\tilde{p}_i +1}_i} = T^{-1}  ( \frac{A_i}{
    x^{p_i+1}_i} T - \pder{T}{x_i}), \quad 1 \leq i \leq n.
\end{equation}
We say that system~\eqref{eq:equiv} is \textit{equivalent} to
system~\eqref{eq:sys} and we write $T[A]:=[\tilde{A}]$. It can be easily
verified that complete integrability is inherited by an equivalent
system. Subsequently, to stay in the same class of systems under study, special
care will be taken so that the transformations used in our considerations do not
alter the normal crossings. In fact, a major difficulty within the symbolic
manipulation of system~\eqref{eq:sys} arises from~\eqref{eq:gauge}. It is
evident that any transformation alters all the components simultaneously. In
particular, the equivalent system does not necessarily inherit the normal
crossings even for very simple examples.
\begin{example}\cite[Section 4]{key5}
\label{exmnaive}
Consider the following completely integrable Pfaffian system with normal
crossings of Poincar\'e rank $(3,1)$:
\begin{equation*}
\begin{cases}
  x_1^4 \pder{F}{x_1} = A_{1}(x_1, x_2)\; F = \left(\begin{matrix} x_1^3 +x_2 &
      x_2^2 \\ -1 & -x_2 + x_1^3 \end{matrix}\right)  F,\\
  x_2^2 \pder{F}{x_2} = A_{2}(x_1, x_2)\; F = \left(\begin{matrix} x_2 & x_2^2
      \\ -2 & -3x_2\end{matrix}\right) F.
\end{cases}
\end{equation*}
This system appears within the reduction of the system of Example~\ref{ex:sim}
in the introduction. As we have seen, there exists a transformation which
drops $p_1$ to zero. This can also be attained by the transformation
\[F = \left(\begin{matrix} x_1^3 & -x_2^2 \\ 0 & x_2 \end{matrix}\right) G,\]
which is computed by the univariate-case Moser-based rank reduction algorithm,
upon regarding the first component as an ODS in $x_1$ and $x_2$ as a
transcendental constant. This results in the equivalent system
\[
  \begin{cases}
    \label{gaugePfaffian}
    x_1x_2 \pder{G}{x_1} = \tilde{A}_{1} (x_1,x_2) \; G= \left(\begin{matrix} -2
        x_2 & 0 \\ -1 & x_2 \end{matrix}\right)  G,\\
    x_2^3 \pder{G}{x_2}= \tilde{A}_{2}(x_1,x_2) \; G = \left(\begin{matrix}
        -x_2^2 & 0 \\ -2 x_1^3 & -2 x_2^2 \end{matrix}\right) G.
\end{cases}
\]
We can see that such a transformation achieves the goal of reducing the
Poincar\'e rank of the first component. However, it alters the normal crossings
as it introduces the factor~$x_2$ on the left hand side of the first
component. Moreover, it elevates the Poincar\'e rank of the second component.
\end{example}
In order to preserve the normal crossings, we restrict the class of
transformations that we use in  our algorithm:
 \begin{definition}
\label{wcompatible}
Let $T\in GL_d({\rm K})$. We say that the transformation $F= TG$ (respectively~$T$) is \textit{weakly compatible} with system $[A]$ if $T[A] :=\tilde{A}$ is
again a completely integrable Pfaffian system with normal crossings. In
particular, $\tilde{A}_i \in {\rm R}^{d \times d}$ for every
$i \in \{1, \dots, n\}.$
\end{definition}
Clearly, any constant or unimodular invertible matrix is an example of such
transformations. 

In the sequel, we will also need to resort to transformations with stronger
properties:
 \begin{definition}
\label{compatible}
Let $T\in GL_d({\rm K})$. We say that the transformation $F= TG$
(respectively~$T$) is \textit{compatible} with system $[A]$ if it is weakly
compatible with $[A]$ and the Poincar\'{e} rank of each individual component of
$T[A]$ does not exceed that of the respective component of $[A]$.
\end{definition}
\subsection{Fundamental Matrix of Formal Solutions}
Before studying how to construct formal solutions to a given system, the
question arises if and how many solutions exist. The language of stable modules
over the ring of power series is used in~\cite[Theorem 1]{key53} and~\cite[Main
Theorem]{key4} independently to establish the following theorem which gives an
answer to this question.

\begin{theorem}
  \label{gerardexistence}
  There exist strictly positive integers~$\alpha_i$, $1 \leq i \leq n$, and an
  invertible matrix $T \in {\rm R}^{d\times d}_t$ such that, upon setting
  $x_i = t_i^{\alpha_i}$, the transformation $T(t)$ yields the following
  equivalent system:
  \[
    t_i^{\alpha_i \tilde{p}_i+1} \pder{G}{t_i}= \tilde{A}_{i}(t_i) G, \quad 1 \leq
  i \leq n,\]
 where
 \[ \tilde{A}_{i}(t_i) = \operatorname{Diag} (\tilde{A}^{11}_{i}(t_i),
   \tilde{A}^{22}_{i}(t_i), \dots, \tilde{A}^{jj}_{i}(t_i)) ,\]
 and for every $\ell \in \{ 1, \dots, j\}$ we have that
 $\tilde{A}^{\ell \ell}_i(t_i) $ is a square matrix of dimension $d_\ell$ of the
 form
  \[ \tilde{A}^{ \ell \ell}_{i}(t_i) = w^{ \ell \ell}_{i} (t_i) I_{d_\ell} +
    t_i^{\alpha_i \tilde{p}_i}( c^{ \ell \ell}_{i}  I_{d_\ell} + N^{ \ell \ell}_{i}),\]
where
\begin{itemize}
\item $d_1 + d_2 + \dots + d_j = d$;
\item
  $w^{ \ell \ell}_{i} (t_i) = \sum_{m=0}^{\alpha_i \tilde{p}_i -1} \lambda_{im \ell}
  t_i^{m}$ is a polynomial in $t_i$, with coefficients in
  $\mathbb{C}$;
\item $c^{ \ell \ell}_{i} \in \mathbb{C}$ and $N^{ \ell \ell}_{i}$ is a constant
  (with respect to all derivations $\partial/\partial t_i$) $d_\ell$-square
  matrix having nilpotent upper triangular form;
\item for any fixed $\ell \in \{ 1, \dots, j \}$, the matrices
  ${\{ N^{ \ell \ell}_{i}\}}_{i = 1, \dots, n}$ are permutables;
\item for all $\ell \in \{ 1, \dots, j-1 \}$, there exists
  $i \in \{ 1, \dots, n\}$ such that
 \[ w^{ \ell \ell}_{i} (t_i) \neq w^{(\ell+1)(\ell+1)}_{i}
    (t_i) \text{\qquad or\qquad}  c^{ \ell \ell}_{i} - c^{(\ell+1)(\ell+1)}_{i} \not
    \in \mathbb{Z}.\]
\end{itemize}
\end{theorem}

This theorem guarantees the existence of a transformation which takes
system~\eqref{eq:sys} to the so-called Hukuhara-Turrittin's normal form from
which the construction of a fundamental matrix of formal
solutions~\eqref{eq:sol} is straightforward.  In fact, we have:
\begin{corollary}
  \label{gerardsol}
  Given system~\eqref{eq:sys}, a fundamental matrix of formal solutions exists
  and is of the form
  \begin{equation}
    \label{eq:sol}
    \Phi(x_1^{1/s_1}, \dots, x_n^{1/s_n} ) \prod_{i=1}^{n} x_{i}^{{C}_i}
    \prod_{i=1}^{n} \operatorname{exp}(Q_i(x_i^{-1/s_i})),
  \end{equation}
  where $\Phi$ is an invertible matrix with entries in ${\rm R}_t$ and for each
  $i \in \{1, \dots, n\}$ we have:
  \begin{itemize}
  \item $s_i$ is a positive integer;
  \item the diagonal matrix
    \[
    Q_i(x_i^{-1/s_i}) = \operatorname{Diag}\left(q_{i,1}(x_i^{-1/s_i}),
      q_{i,2}(x_i^{-1/s_i}), \dots, q_{i,d}(x_i^{-1/s_i})\right)
    \]
    contains polynomials in $x_i^{-1/s_i}$ over $\set C$
    without constant terms. We refer to $Q_i(x_i^{-1/s_i})$ as the
    $x_i$-exponential part. Under the notations of
    Theorem~\ref{gerardexistence}, it is obtained by formally integrating
    $\frac{w_{i}^{\ell\ell}}{t_i^{\alpha_i \tilde{p}_i +1}}$;
  \item $C_i $ is a constant matrix which commutes with
    $Q_i(x_i^{-1/s_i})$.
  \end{itemize}
\end{corollary}
A singular system $[A]$ is said to be \textit{regular singular} whenever, for every
$i \in \{1, \dots , d \}$, $Q_i(x_i^{-1/s_i})$ is a zero matrix.  Otherwise,
system~\eqref{eq:sys} is said to be \textit{irregular singular} and the entries
of $Q_i(x_i^{-1/s_i})$, $1\leq i \leq n$, determine the main asymptotic behavior
of the actual solutions as $x_i \rightarrow 0$ in appropriately small sectorial
regions~\cite[Proposition 5.2, pp 232, and Section~4]{key1}.
\begin{definition}
\label{katzpfaff}
Let $i \in \{1, \dots, n\}$. If $Q_i(x_i^{-1/s_i})$ is a nonzero matrix then we set $m_{i,j}$ to be the minimum order in $x_i$ within the terms
of
$q_{i,j}(x_i^{-1/s_i})$ for $1\leq j\leq d$. The $x_i$-\textit{formal exponential growth order}
($x_i$-exponential order, in short) of~$A_i$ is the rational number
$$\omega(A_i) = - \operatorname{min}_{1 \leq j \leq d}  m_{i,j}.$$
The $n$-tuple of rational numbers
$\omega(A) = (\omega(A_1), \dots, \omega(A_n))$ then defines the exponential
order of system~$[A]$. Otherwise, we set $\omega(A_i) =0$.
\end{definition}
If two systems are equivalent then they have the same $x_i$-exponential parts,
and consequently the same $x_i$ exponential orders, for all $1 \leq i \leq n$,
under any transformation $T \in GL_d({\rm K})$.
\begin{example}[Example~\ref{ex:sim} cont.]
  From our investigations in the example of Section~\ref{sec:intro}, we see that
  for the given fundamental system of formal solutions, we have non-zero exponential
  parts with $\omega(A_1)=1$ and $\omega(A_2)=2$ and so the system is irregular
  singular (although $s_1 = s_2 =1$).
\end{example}
The above theoretical results on existence do not establish the formal reduction
itself, that is the algorithmic procedure which computes explicitly the
$\alpha_i$'s and a transformation which takes the system to a normal form that
allows the construction of such solutions. This will be our interest in the
following sections.

The computation of the formal invariants is a difficult task in the univariate
case~\cite{key24,B10}. However, we will prove in Section~\ref{sec:invariants} that
in the multivariate case, these invariants can be computed from associated
univariate systems. Unlike the univariate case, the main difficulties of the
algorithm lie in rank reduction. Before proceeding to describe the algorithms we
propose, we give a property of the transformations which can be deployed:
\begin{proposition}
\label{Tform}
Consider a completely integrable Pfaffian system $[A]$ with normal
crossings. Let $T \in GL_d({\rm K})$ and set $T[A] = [\tilde{A}]$. If $T$ is a
transformation which is weakly compatible with system $[A]$ then
$T \in GL_d(R_L)$.
\end{proposition}
\begin{proof}
It follows from~\eqref{eq:gauge} that
$$
 \pder{T}{x_i}  = \frac{A_i}{
    x^{p_i+1}_i} T - T \frac{\tilde{A}_i }{x^{\tilde{p}_i +1}_i} , \quad 1 \leq i \leq n.
$$
Thus, we have (see, e.g.~\cite[Proposition 1, proof, pp 6]{key6}):
\begin{equation}
\label{Tform11}
\pder {\det(T)}{x_i} = ( \frac{\operatorname{tr}(A_i)}{x^{p_i+1}} -
\frac{\operatorname{tr}(\tilde{A}_i)}{x^{\tilde{p}_i+1}}) \det(T) , \quad 1 \leq
i \leq n.
\end{equation}
Therefore, $\det(T)$ itself is a solution of a completely integrable Pfaffian
system with normal crossings. By Corollary~\ref{gerardsol}, $\det(T)$ has the
form~\eqref{eq:sol}. Since $T \in {\rm K}^{d \times d}$ then $\det(T)$ is free
of logarithmic and exponential terms. Hence, $\det(T)$ corresponds to a log-free
regular solution of~\eqref{Tform11}. Thus, $\det(T) \in {\rm R}_L$. The same
argument serves to prove that $\det(T^{-1})$ is an element of ${\rm R}_L$ as
well, upon remarking that $T^{-1}[\tilde{A}] = [A]$. Hence,
$\det(T)^{-1} \in {\rm R}_L$ and consequently $T \in R_L^{d \times d}$.
\end{proof}
However, the converse of Proposition~\eqref{Tform} is not true, which
complicates the task of constructing adequate transformations in the reduction
process (see, e.g., Example~\ref{exmnaive} or the shearing transformations of
Section~\ref{shearpfaff}).

\section{Structure of the Main Algorithm}
\label{sec:outline}
If one is only interested in the asymptotic behavior of the
solutions of system~$[A]$, then one can compute the formal invariants from
associated univariate systems as we prove in Section~\ref{sec:invariants}.
\begin{figure}
\label{fig1}
\centering
\begin{tikzpicture}[scale=0.6, every node/.style={scale=0.6}]
\draw [line width=1pt]
  (0,0) rectangle (2.2,2.2) node[label={[align=center]Input\\system}] at
  (1.1,0.4){ };

\path[->] (2.3,2.3) edge (2.9,2.9);
\path[->] (2.3,-0.1) edge (2.9,-0.7);

\draw [line width=1pt]
  (3,3) rectangle (5.2,5.2) node[label={[align=center]First\\component}] at
  (4.1,3.4){ };

\path[->] (5.3,5.2) edge (5.9,5.8);
\path[->] (5.3,4.05) edge (5.9,4.05);
\path[->] (5.3,2.9) edge (5.9,2.3);

\draw[thick,dotted] (4.1,2.5) -- (4.1,-0.3);

\draw [line width=1pt]
  (3,-3) rectangle (5.2,-0.8) node[label={[align=center]Last\\component}] at
  (4.1,-2.6){ };

\path[->] (5.3,-1.9) edge (5.9,-1.4);
\path[->] (5.3,-1.9) edge (5.9,-2.5);

\draw [line width=1pt] (6,-0.8) rectangle (8.2,-1.8)
node[label={[align=center]System w.\\lower dim.}] at (7.1,-1.9){ };

\draw [line width=1pt] (6,-2.0) rectangle (8.2,-3.0)
node[label={[align=center]System w.\\lower dim.}] at (7.1,-3.1){ };

\draw (8.3,-2.5) -- (8.6,-2.5) -- (8.6,-3.6) -- (1.1,-3.6);
\draw (8.3,-1.3) -- (8.6,-1.3) -- (8.6,-3.6);
\path[->] (1.1,-3.6) edge (1.1,-0.1);

\draw [line width=1pt] (6,5.4) rectangle (8.2,7.6)
node[label={[align=center]$\geq$ 2 distinct\\eigenvalues}] at (7.1,5.8){ };

\path[->] (8.3,6.5) edge (8.9,7.1);
\path[->] (8.3,6.5) edge (8.9,5.9);

\draw [line width=1pt]
  (6,3) rectangle (8.2,5.2) node[label={[align=center]Unique\\eigenvalue}] at
  (7.1,3.4){ };

\draw (8.3,4.1) -- (8.7,4.1) -- (8.7,2.4); 
\path[->] (8.7,2.4) edge (8.3,2.4);

\draw [line width=1pt]
  (6,0.6) rectangle (8.2,2.8) node[label={[align=center]Nilpotent}] at
  (7.1,1.2){ };

\path[->] (8.3,1.0) edge (8.9,1.0);

\draw [line width=1pt]
  (9,6.6) rectangle (11.2,7.6) node[label={[align=center]System w.\\lower dim.}] at
  (10.1,6.5){ };

  \draw [line width=1pt] (9,5.4) rectangle (11.2,6.4)
  node[label={[align=center]System w.\\lower dim.}] at (10.1,5.3){ };

\draw (11.3,7.1) -- (11.6,7.1) -- (11.6,8.1) -- (1.1,8.1); 
\draw (11.3,5.9) -- (11.6,5.9) -- (11.6,8.1); 
\path[->] (1.1,8.1) edge (1.1,2.3);

\draw [line width=1pt]
  (9,0) rectangle (11.2,2.2) node[label={[align=center]Apply rank\\reduction\\in
  first var.}] at
  (10.1,0.25){ };

\path[->] (11.3,1.1) edge (11.9,2.0);
\path[->] (11.3,1.1) edge (11.9,0.2);

\draw [line width=1pt] (12,1.2) rectangle (14.2,3.4)
node[label={[align=center]$\geq$ 2 distinct\\eigenvalues}] at (13.1,1.6){ };

\draw (13.1,3.5) -- (13.1,8.5) -- (7.1,8.5);
\path[->] (7.1,8.5) edge (7.1,7.7);

\draw [line width=1pt]
  (12,1.0) rectangle (14.2,-1.2) node[label={[align=center]Nilpotent}] at
  (13.1,-0.6){ };

\path[->] (14.3,-0.1) edge (14.9,-0.1);

\draw [line width=1pt]
  (15,1.0) rectangle (17.2,-1.2) node[label={[align=center]Compute\\ exp. order
    \\in first var.}] at
  (16.1,-0.95){ };

\path[->] (17.3,-0.1) edge (17.9,-0.1);

\draw [line width=1pt]
  (18,1.0) rectangle (20.2,-1.2) node[label={[align=center]Apply rami-\\
    fication in\\ first var.}] at
  (19.1,-0.95){ };

\draw (19.1,-1.3) -- (19.1,-1.9) -- (10.1,-1.9);
\path[->] (10.1,-1.9) edge (10.1,-0.1);
\end{tikzpicture}
\caption{Computing a fundamental matrix of formal solutions by working with one of
  the components, e.g. the first component. The other components would follow
  the chosen component in the uncoupling.}
\end{figure}
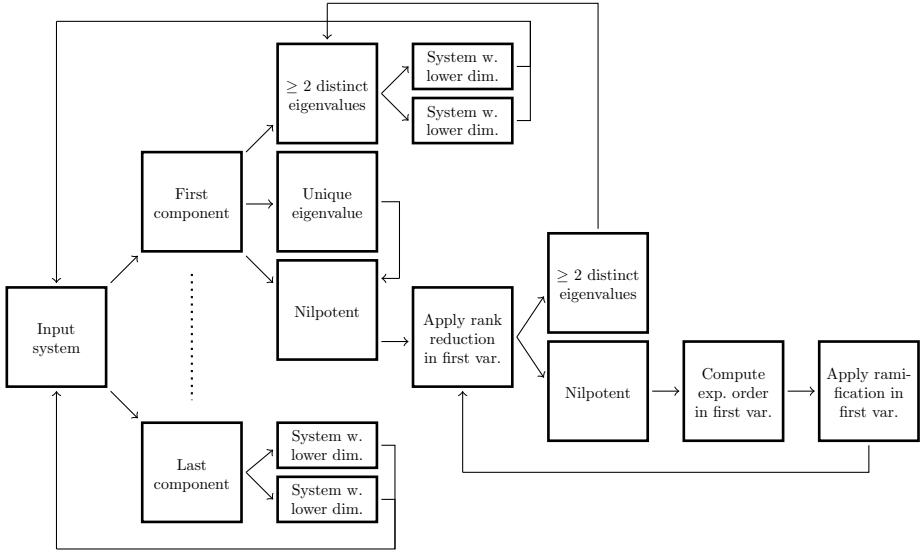
If the singularity is regular or one is interested in computing a full
fundamental matrix of formal solutions, as given by~\eqref{eq:sol}, then, besides
computing these invariants, further involved steps are required, as illustrated
in Example~\ref{ex:sim}. The recursive algorithm we propose generalizes that of
the univariate case given by the first author in~\cite{key24}. At each level of
recursion with input $[A]$, we consider the leading matrix coefficients
$A_{i,0}=A_i(x_i=0)$ (we use both notation interchangeably) and distinguish
between three main cases:

\begin{enumerate}
\item There exists at least one index $i \in \{ 1, \dots, n \}$ such that
  $A_{i,0}$ has at least two distinct eigenvalues.
\item All of the leading matrix coefficients have exactly one eigenvalue and
  there exists at least one index $i \in \{ 1, \dots, n \}$ such that
  $A_{i,0}$ has a nonzero eigenvalue.
\item For all $i \in \{ 1, \dots, n \}$, $A_{i,0}$ is nilpotent.
\end{enumerate}
\vspace{0.3cm}

In order to identify the properties of the eigenvalues of $A_{i}(x_i=0)$, it
suffices to consider the constant matrix $A_i(x=0)$ due to the following
well-known proposition (see, e.g.,~\cite[Proposition 1, pp 8]{key4}
or~\cite[Proposition 2.2]{key9} for a proof within the context of eigenrings):
\begin{proposition}
\label{constev}
The eigenvalues of $A_{i,0}$ , $1 \leq i \leq n$, belong to $\mathbb{C}$.
\end{proposition}
Then, based on the above classification, a linear or an exponential
transformation will be computed as described in the following subsections.

\subsection{Distinct Eigenvalues: Uncoupling the System Into Systems 
of Lower Dimensions}
\label{diagpfaff}
Whenever there exists an index $i \in\{1,\dots,n \}$ such that $A_{i,0}$ has at
least two distinct eigenvalues, the system can be uncoupled into subsystems of
lower dimensions as shown in Theorem~\ref{blockpfaff}. For a constructive proof,
one may refer to~\cite[Section 5.2, pp 233]{key1}.
\begin{theorem}\label{blockpfaff}
  Suppose that for some $i \in\{ 1,\dots,n \}$, the leading matrix coefficient
  $A_{i,0}$ has at least two distinct eigenvalues. Then there exists a unique
  transformation $T \in GL_d({\rm R})$ of the form
  \[T{(x)} =\left(
    \begin{matrix}
      T^{11} &   T^{12} \\
      T^{21} & T^{22}\\
    \end{matrix}\right)=\left(
    \begin{matrix}
      I_{d'} &   T^{12}{(x)} \\
      T^{21}{(x)} & I_{d-d'}\\
    \end{matrix}\right),
  \]
where $0< d'< d$, such that the transformation $F = T G$ yields
  the equivalent system
  \[x_{i}^{p_{i} +1} \pder{G}{x_i} = \left(
    \begin{matrix}
      \tilde{A}^{11}_{i}(x) & O\\
      O &     \tilde{A}^{22}_{i}(x)\\
    \end{matrix}\right)G, \quad 1 \leq i \leq n .
  \]
  and $\tilde{A}^{11}_{i}(x),\tilde{A}^{22}_{i} (x)$, $i \in \{ 1, \dots, n\}$
  are of dimensions $d'$ and $d-d'$ respectively. 
\end{theorem}
The theorem can be restated by saying that if one of the components of the
system has a leading matrix coefficient with at least two distinct eigenvalues,
then it can be uncoupled. All of the other components are uncoupled
simultaneously. In the sequel, we aim to determine changes of the independent
variables $x_i$ (ramifications), and construct transformations, which will allow
the reduction of any input system to a system for which the leading matrix
coefficient of at least one of its components has at least two distinct
eigenvalues. This allows us to either arrive at a system with lower Poincar\'e
rank or uncouple it into several subsystems of lower dimensions. The
recursion stops whenever we arrive at regular singular ($p =(0, \dots, 0)$) or scalar
($d=1$) subsystems. The former have been already investigated in~\cite[Chapter
3]{key73} and the resolution of the latter is straightforward.

We remark that, by Proposition~\ref{constev}, it suffices that there exists
$i \in \{1 , \dots, n\}$ such that the constant matrix
$A_{i}(x_1=0, \dots, x_{i}=0, \dots, x_n=0)$ has at least two distinct
eigenvalues.

\subsection{Unique Eigenvalue: Shifting}
\label{shiftpfaff}
For any $i\in\{1,\dots,n\}$ such that $A_{i,0}$ has a unique nonzero
eigenvalue $\gamma_{i} \in \mathbb{C}$, applying the so-called eigenvalue
shifting
\begin{equation*}
F = \operatorname{exp}\left(\int^{\text{\rlap{$x_{i}$}}}\gamma_{i}
  z_{i}^{-p_{i}-1} dz_{i}\right)  G,
\end{equation*}
yields a system $[\tilde{A}]$ whose ${i}^{th}$ component has a nilpotent leading
matrix coefficient:
\begin{eqnarray*}
  x_i^{\tilde{p}_i+1} \pder{G}{x_i} = \tilde{A}_i{(x)} G, \quad \text{where}
\quad \tilde{A}_{i}{(x)} = A_{i}{(x)} - \gamma_{i} I_d .
\end{eqnarray*}
The other components of the system are not modified by this transformation which
is clearly compatible with system $[A]$.

Hence, due to the uncoupling and shifting, we can assume without
loss of generality that for all $i \in \{ 1, \dots, n\}$, the leading matrix
coefficients $A_{i,0}$ are nilpotent.
\subsection{Nilpotency: Rank Reduction and Exponential Order}
\label{rampfaff}
In the univariate case, $n=1$, the nilpotency of $A_{1,0}$ suggests at least one
of the following two steps, as proposed by the first author in~\cite{key24}:
Rank reduction and computation of the exponential order $\omega(A_1)$. The
former reduces $p_1$ to its minimal integer value. It is possible that $p_1$
drops to zero, i.e.\ we arrive at a regular singular system, or that the leading matrix
coefficient of the resulting system has at least two distinct eigenvalues, in
which case we can again uncouple the system. Otherwise,
$\omega(A_1) = \ell / m$ is to be computed, where~$\ell$ and $m$ are
coprime. Then, by setting $x_1 = t_1^{m}$ and applying rank reduction again, it
is proven that we arrive at a system whose leading matrix coefficient has two
distinct eigenvalues. Therefore, the system can be uncoupled (see
Figure~1).

The bivariate case, $n=2$, is studied by the first and third authors of this
paper in~\cite{key101}. For rank reduction, the properties of principal ideal
domains were used. To determine the formal exponential order $\omega(A)$,
associated univariate systems were defined. In this paper, we show that on the
one hand, this approach to determine the formal exponential order remains valid
in the multivariate setting, as we will see in the next section. On the other
hand, the generalization of the rank reduction algorithm to the multivariate
case is nontrivial and is discussed in Section~\ref{sec:rankred}. The
multivariate formal reduction algorithm is then summed up in
Section~\ref{mainpfaff}.

\section{Computing the Formal Invariants}
\label{sec:invariants}
In the univariate case, where the system is given by a single matrix $A_1$,
$\omega(A_1)$ can be computed from the characteristic polynomial of $A_1$, i.e.\
$\det(\lambda I_d - A_1)$, based on the analysis of a Newton polygon associated
with the system~\cite[Theorem 1]{key24}. In this section we show that one need
not search for a generalization of this algorithm to the multivariate case as
the formal invariants of $[A]$, i.e.\ the exponential parts and $\omega(A)$, can
be obtained from an associated univariate system. We do not only give a method
to retrieve these invariants but we also reduce computations to computations
with univariate rather than multivariate formal series.
\begin{definition}
\label{defass}
Given a Pfaffian system $[A]$, we call the following the \textit{associated ODS}
of $[A]$:
  \begin{equation*}
  x_i^{p_i+1} \frac{d}{dx_i} \mathcal{F}_i =
\mathcal{A}_{i}(x_i)\;\mathcal{F}_i, \quad 1 \leq i \leq n,
  \end{equation*}
  where
  $\quad \mathcal{A}_{i}(x_i) := A (x_1=0, \dots, x_{i-1}=0, x_i, x_{i+1}=0,
  \dots, x_n=0).$
\end{definition}
\begin{theorem}
\label{exponentialpfaff}
For every $i \in \{ 1, \dots, n \}$, the $x_i$-exponential part of a
Pfaffian system is equal to the exponential part of the $i^{th}$
component of its associated ODS.
 \end{theorem}
 To establish this result, we rely on a triangular form weaker than the
 Hukuhara-Turrittin's normal form given in Theorem~\ref{gerardexistence}. This
 weaker form suffices to give insight into the computation of~\eqref{eq:sol}.

 The following theorem is an reformulation of a theorem which was first given
 in~\cite[Proposition 3, pp 654]{key3} for the bivariate case, and then
 generalized in~\cite[Theorem 2.3]{key53} to the general multivariate case.
\begin{figure}
\label{fig2}
\centering
\begin{tikzpicture}[scale=0.6, every node/.style={scale=0.6}]
\draw [line width=1pt]
  (0,0) rectangle (1.8,1.8) node[label={[align=center]Input\\system}] at
  (0.9,0.2){ };

\path[->] (1.9,1.9) edge (2.9,2.9);
\path[->] (1.9,-0.1) edge (2.9,-1.1);

\draw [line width=1pt]
  (3,3) rectangle (4.8,4.8) node[label={[align=center]First\\component}] at
  (3.9,3.2){ };

\path[->] (4.9,3.9) edge (5.9,3.9);

\draw[thick,dotted] (3.9,2.5) -- (3.9,-0.7);

\draw [line width=1pt]
  (3,-3) rectangle (4.8,-1.2) node[label={[align=center]Last\\component}] at
  (3.9,-2.8){ };

\path[->] (4.9,-2.1) edge (5.9,-2.1);

\draw [line width=1pt] (6,3) rectangle (7.8,4.8)
node[label={[align=center]First\\associated\\ ODS}] at (6.9,3.1){ };

\path[->] (7.9,3.9) edge (8.9,3.9);

\draw[thick,dotted] (6.9,2.5) -- (6.9,-0.7);

\draw [line width=1pt] (6,-3) rectangle (7.8,-1.2)
node[label={[align=center]Last\\associated\\ ODS}] at (6.9,-2.9){ };

\path[->] (7.9,-2.1) edge (8.9,-2.1);

\draw [line width=1pt] (9,3) rectangle (10.8,4.8)
node[label={[align=center]Exp. part\\in first var.}] at (9.9,3.3){ };

\draw[thick,dotted] (9.9,2.5) -- (9.9,-0.7);

\draw [line width=1pt] (9,-3) rectangle (10.8,-1.2)
node[label={[align=center]Exp. part\\in last var.}] at (9.9,-2.7){ };

\path[->] (10.9,2.9) edge (11.9,1.9);
\path[->] (10.9,-1.1) edge (11.9,-0.1);

\draw [line width=1pt]
  (12,0) rectangle (13.8,1.8) node[label={[align=center]Exp. part}] at
  (12.9,0.4){ };
\end{tikzpicture}
\caption{Computing the exponential part from associated ODS's}
\end{figure}
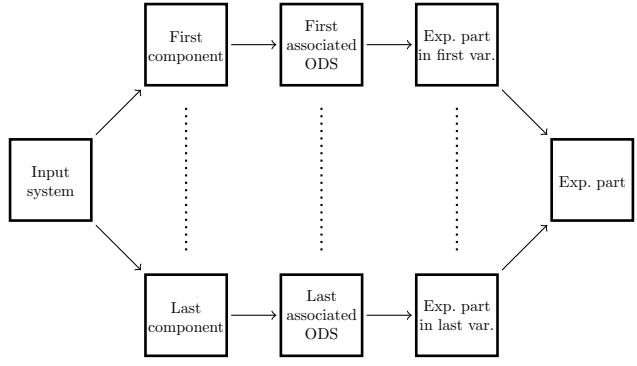
\begin{theorem}
\label{gerardtransformation}
Consider the Pfaffian system $[A]$. There exists a positive integer $\alpha_1$,
and a transformation $T \in GL_d({\rm K_t})$ (where $x_1 = t_1^{\alpha_1}$ and
$x_i = t_i, \; 2 \leq i \leq n$), such that the transformation $F= T G$ yields
the equivalent system:
\begin{equation}
  \label{eq:ger}
  \begin{cases}
    t_1^{\alpha_1 \hat{p}_1 + 1} \pder{G}{t_1} = \hat{A}_{1}(t_1, x_2, \dots, x_n)
\; G, \\[7pt] x_i^{ \hat{p}_i+1} \pder{G}{x_i}= \hat{A}_{i}(x_2, \dots, x_n)\; G,
\quad 2 \leq i \leq n,
\end{cases}
\end{equation}
where
\begin{alignat*}2 & {\tilde{A}}_{1}(t_1, x_2, \dots, x_n) &\; =\;
  &\operatorname{Diag} (\hat{A}^{11}_{1}, \hat{A}^{22}_{1},
  \dots, \hat{A}^{jj}_{1}), \\
  &\hat{A}_{i}(x_2, \dots, x_n)& =\; & \operatorname{Diag}
  (\hat{A}^{11}_{i}, \hat{A}^{22}_{i}, \dots, \hat{A}^{jj}_{i}), \quad 2
  \leq i \leq n,
\end{alignat*}
and for all $\ell \in \{1, \dots, j\}$ and $i \in \{2, \dots, n\}$ the entries
of $\hat{A}^{{\ell} {\ell} }_{{i}}$ lie in ${\rm R}_{\bar{x}_{1}}$. The
$\hat{A}_1^{\ell\ell}$'s are of the form
 \[\hat{A}^{{\ell} {\ell} }_{1} = w^{{\ell} {\ell} }_{1} (t_1) I_{d_{\ell}
   } + t_1^{\alpha_1 \hat{p}_1}(\hat{N}^{{\ell} {\ell} }_{1}(x_2, \dots, x_n) +
   c^{{\ell} {\ell} }_{1} I_{d_{\ell}}),\] where

\begin{itemize}
\item $d_1 + d_2 + \dots + d_j = d$;
\item $w^{{\ell} {\ell} }_{1} (t_1)$ and $c^{{\ell} {\ell} }_{1}$ are as in
  Theorem~\ref{gerardexistence};
\item If $\ell, \ell' \in \{ 1, \dots, j-1 \}$ and $\ell \neq \ell'$, then
  $ w^{ \ell \ell}_{1} (t_1) \neq w^{ \ell' \ell'}_{1} (t_1)$ or
  $c^{ \ell \ell}_{1} - c^{ \ell' \ell'}_{1} \not \in
  \mathbb{Z}$;
\item $\hat{N}^{{\ell} {\ell} }_{1} (x_2, \dots, x_n)$ is a nilpotent
  $d_{\ell} $-square matrix whose entries lie in ${\rm R}_{\bar{x}_{1}}$.
\end{itemize}

\noindent Moreover, $T$ can be chosen as a product of transformations in
$GL_d({\rm R}_t)$ and transformations of the form
$\operatorname{Diag} (t_1^{\beta_1}, \dots, t_1^{\beta_d})$, where
$\beta_1, \dots, \beta_d$ are non-negative integers.$\hfill\qed$
\end{theorem}

\begin{proof}[Proof of Theorem~\ref{exponentialpfaff}] Upon the change of
  independent variable $x_1=t_1^{\alpha_1}$, the transformation $F = T G$ yields
  system~\eqref{eq:ger} for which the first component is given by
\begin{equation*}
t_1^{\alpha_1 \hat{p}_1 + 1} \pder{G}{t_1} = \tilde{A}_{1}(t_1, x_2, \dots, x_n)  \; G.
\end{equation*}
with the notations and properties as in Theorem~\ref{gerardtransformation}. It
then follows from~\eqref{eq:gauge} that
\begin{equation}
\label{relation}
t_1^{\alpha_1 {\hat{p}}_1 + 1}\pder{T}{t_1}= \alpha_1
A_{1}(x_1=t_1^{\alpha_1})\; T - T {\hat{A}}_{1}.
\end{equation}
Due to the particular choice of $T$ in Theorem~\ref{gerardtransformation}, we
can set $x_i=0$, $2 \leq i \leq n$ in~\eqref{relation}. In particular, the
relation between the leading terms
\begin{alignat*}2 &\mathcal{A}_{1}(x_1=t_1^{\alpha_1}) &\;:=\; &A_{1} (x_1 =
  t_1^{\alpha_1}, x_2=0,\dots,x_n=0),\\ &\hat{\mathcal{A}}_{1} (t_1) &\; :=
  \;&\hat{A}_{1} (t_1, x_2=0, \dots, x_n=0),\\ &\mathcal{T} (t_1) &\; :=
  \;&T(t_1, x_2=0,\dots,x_n=0),
\end{alignat*}
is given by
\begin{equation*}
  t_1^{\alpha_1 \hat{p}_1 + 1} \pder{\mathcal{T}}{t_1} = \alpha_1
{\mathcal{A}}_{1}(x_1=t_1^{\alpha_1})\; \mathcal{T} - \mathcal{T}
{\hat{\mathcal{A}}}_{1}.
\end{equation*}
Hence, the systems given by $\alpha_1{\mathcal{A}}_{1}(x_1=t_1^{\alpha_1})$
(respectively ${\mathcal{A}}_{1}$) and ${\hat{\mathcal{A}}}_{1}$ are
equivalent. It follows that they have the same formal invariants.  Clearly, the
same result can be obtained for any of the other components via permutation with
the first component. Noting that the $x_i$-exponential part is independent of
$\bar{x}_i$ completes the proof.
 \end{proof}
 For univariate systems, the \textit{true Poincar\'e rank} $p_{true}(A_1)$ is
 defined as the smallest integer greater or equal than the exponential order  $\omega(A_1)$ of 
 the system $[A_1]$. It is known that this integer coincides with the minimal value for
 $p_1$ which can be obtained upon applying any  non-ramified linear transformation 
 to~$A_1$. With the help of Theorem~\ref{exponentialpfaff} we can establish the
 analogous result for multivariate systems. We first give the following
 definition:
\begin{definition}
  Let $[A]$ be a Pfaffian system and for any $i\in\{1,\dots,n\}$, let
  $p_{true}(A_i)$ be the minimal integer value which bounds the exponential order
  in $x_i$, i.e.
  \[
  p_{true}(A_i) - 1 < \omega(A_i)  \leq p_{true}(A_i).
  \]
  Then $p_{true}(A)=(p_{true}(A_1),\dots,p_{true}(A_n))$ is
  called the true Poincar\'e rank of $A$.
\end{definition}
It is shown by Deligne and van den Essen separately in~\cite{key20,key21}, in
the multivariate setting that a necessary and sufficient condition for system
$[A]$ to be regular singular is that each individual component $A_i$, considered
as a system of ordinary differential equations in $x_i$, with the remaining
variables held as transcendental constants, is regular singular. As a
consequence, system $[A]$ is regular singular if and only if its true
Poincar\'{e} rank is $(0,0, \dots,0)$. To test this regularity, algorithms
available for the univariate case of $n=1$ (e.g.~\cite{key26,key57}) can be
applied separately to each of the individual components. The following corollary
follows directly from Theorem~\ref{exponentialpfaff}, showing that the $i^{th}$
component of the true Poincar\'e rank of system $[A]$ is equal to the  true
Poincar\'e rank of the $i^{th}$ associated (univariate) ODS.
\begin{corollary}
\label{prank}
For all $1\leq i\leq n$ we have
\[
  p_{true}(A_i) =p_{true}(\mathcal{A}_i).
\]
\end{corollary}
\begin{proof}
  For the proof, it suffices to remark that\footnote{Stronger bounds are given
    in~\cite[Remark 3]{key24}}
  \[p_{true}({\mathcal{A}}_{i}) -1 < \omega(A_i) = \omega({\mathcal{A}}_{i} ) \leq
    p_{true}({\mathcal{A}}_{i}).\qedhere\]
\end{proof}
From this Corollary it does not yet follow for the multivariate case, as in the
univariate case, that it is possible to apply a compatible transformation to
system~$[A]$ such that all the $p_i$ simultaneously equal the $p_{true}(A_i)$. We
investigate this possibility in the next section.

In summary, the formal exponential order, the true Poincar\'e rank, and most
importantly the $Q_i$'s in~\eqref{eq:sol}, can be obtained efficiently by
computations with univariate rather than multivariate series, making use of
existing algorithms and packages. As mentioned in the introduction, this
exponential part is of central importance in applications since it determines
the asymptotic behavior of the solution in the neighborhood of an irregular
singularity. To compute a full fundamental matrix of formal solutions, we still
have to determine suitable rank reduction transformations. Transformations which
reduce the rank of the associated systems do not suffice, since they are not
necessarily compatible. We therefore proceed to develop a multivariate rank
reduction algorithm.

\section{Rank Reduction}
\label{sec:rankred}
In this section, we are interested in the rank reduction of Pfaffian systems,
more specifically, the explicit computation of a transformation which, given
system~$[A]$, yields an equivalent system whose Poincar\'e rank is the true
Poincar\'e rank. We show that, under certain conditions, the true Poincar\'e
ranks of the subsystems of~$[A]$ can be attained simultaneously via a
transformation compatible with~$[A]$.

We first generalize Moser's reduction criterion~\cite{key19} to multivariate
systems. We then establish an extension of the algorithm we gave
in~\cite{key101} for the bivariate case to the multivariate setting. The main
problem in the treatment of multivariate systems is that the entries of the
$A_i$ do not necessarily lie in a principal ideal domain. This is a common
problem within the study of systems of functional equations. The same obstacle
arises in~\cite{key5} and in the analogous theory of formal decomposition of
commuting partial linear difference operators established in~\cite{key3239}.

\subsection{Generalized Moser's Criterion}
\label{criterionpfaff}
For univariate systems, Moser's criterion characterizes systems for which rank
reduction is possible. To adapt this criterion to our setting, we
follow~\cite{key41,key19} and define the generalized Moser rank and the Moser
invariant of a system $[A]$ as the following $n$-tuples of rational numbers:
\begin{eqnarray*}
  m (A) = (m (A_1) , \dots, m (A_n)), & \text{where} &  m (A_{i}) =
\max\left(0, p_i + \frac{\operatorname{rank}(A_{i,0})}{d}\right),\\ 
\mu (A) = (\mu (A_1) , \dots, \mu (A_n)), & \text{where} &  \mu (A_{i})
= \operatorname{min}(\{ m (T(A_i)) \mid T \in GL_d({\rm R}_L) \}).
\end{eqnarray*} 
We remark that $\mu(A)$ is well-defined due to Corollary~\ref{prank}.
\begin{definition}
  Consider the partial order $\leq$ on $\set Q^n$ for which $\ell < k$ holds if
  and only if $\ell_j\leq k_j$ for all $1\leq j\leq n$ and there is at least one
  index for which the inequality is strict. The system $[A]$ is called
  reducible if $\mu(A) < m(A)$. Otherwise it is said to be irreducible.
\end{definition}
In other words, system~$[A]$ is irreducible whenever each of its components
is. In particular, it is easy to see from this definition that a system $[A]$ is
regular singular if and only if $\mu(A_i) \leq 1$ for all
$i \in \{ 1, \dots, n\}$, i.e.\ the true Poincar\'e rank is a zero $n$-tuple,
which coincides with Deligne's and van den Essen's criterion for regular
singular systems.

\subsection{Main Theorems}
\label{mainthmpfaff}
With the help of compatible transformations and the criterion established in
Section~\ref{criterionpfaff}, we study the rank reduction of some component
$A_i$ of~$[A]$ which is given by~\eqref{eq:sys}. We will see that rank reduction
can be carried out for each component independently without affecting the
individual Poincar\'e ranks of the other components. We fix
$i \in \{ 1, \dots, n\}$ and we recall that one can expand the components
of~$[A]$ with respect to $x_i$. In particular, we have,
\begin{equation*}
  x_i^{p_i +1 }\pder{F}{x_i} = A_i(x) F = (A_{i,0}({\bar{x}_i})+ A_{i,1}
({\bar{x}_i})x_i + A_{i,2}({\bar{x}_i}) x_i^2 + \dots ) F.
\end{equation*}
 We set
 \[r:= \operatorname{rank} (A_{i,0}).\]
 For all $i$ we can assume without loss of generality that $A_{i,0}$ is not the
 zero matrix and thus the reducibility of system~$[A]$ coincides with the
 existence of an equivalent system such that for some~$i$ the rank of the
 leading matrix coefficient
 $\tilde{A}_{i,0}$ is less than $r$. We establish the
 following theorem:
\begin{theorem}
\label{moserpfaff}
Consider a Pfaffian system~$[A]$ and suppose that $m(A_i) > 1$ for some index
$i \in \{1, \dots, n\}$. If $\mu(A_i) < m(A_i)$ then the polynomial
\begin{equation}
\label{eq:theta}
\theta_{i} (\lambda) := {x_i}^{r} \det(\lambda I + \frac{A_{i,0}}{x_i} + A_{i,1}
)|_{x_i=0}
\end{equation}
vanishes identically in $\lambda$.  
\end{theorem}

\begin{proof}
  Suppose that there exists a transformation $T(x) \in {\rm R}_L^{d \times d}$
  which reduces $m(A_i)$ for some $i \in \{ 1, \dots, n\}$. That is, setting
  $T[A] := [\tilde{A}]$, we have:
  \begin{equation} \label{ppp} \operatorname{m}(\tilde{A}_{i}) <
    \operatorname{m}(A_{i}) .\end{equation} The $i^{th}$ component of
  system $[A]$ can also be viewed as a system of ordinary differential equations
  (ODS) in $x_i$ upon considering the $x_j$'s for $j \in \{1, \dots, n\}$,
  $j \neq i$, to be transcendental constants. Hence, by~\eqref{ppp}
  and~\cite[Theorem 1]{key19}, $\theta_i(\lambda) = 0$.
\end{proof}
Intuitively, the characteristic polynomial of $A_i/x_i$ is used
in~\eqref{eq:theta} to detect the true Poincar\'e rank of the $i^{th}$ component
via the valuation of $x_i$. It turns out that the valuation is only influenced
by $A_{i,0}$ and $A_{i,1}$. Even though the true Poincar\'e rank can be
determined from the associated ODS, the criterion is essential as it leads to
the construction of the transformation~$T$.

The converse of Theorem~\ref{moserpfaff} also hold true under certain conditions
(see Theorem~\ref{moserpfaff2}): If $\theta_i (\lambda)$ vanishes for some index
$i \in \{ 1, \dots, n\}$ then we can construct a compatible transformation
$T \in GL_d({\rm R}_L)$ which reduces $m(A_i)$. We will establish this result
and describe the steps of the algorithm after establishing a series of
intermediate ones. We will need two kinds of transformations, shearing
transformations and column reductions, which we explain in the next two
subsections.

\subsection{Shearing Transformation}
\label{shearpfaff}
Consider the expansion $A_i=\sum_{k=0}^\infty A_{i,k}x_i^k$ of $A_i$ with
respect to $x_i$ for a fixed $i\in\{1,\dots,n\}$. Shearing transformations are
polynomial transformations that, roughly speaking, are used to exchange blocks
between the $A_{i,k}$'s. The ones we consider here are of the form
\[S = \operatorname{Diag} (x_{i}^{\beta_1}, \dots , x_{i}^{\beta_d}),\]
with $\beta_j \in\{0,1\}$ for all $j \in\{1,\dots,d\}$.  We illustrate the
shearing effect of such a transformation in an easy example.

\begin{example}
\label{ex:shearing}
We apply a shearing transformation to a univariate system $[A]$ given by
$x_1^{p_1+1} \frac{\partial}{\partial x_1} F = A_1 F$ where
$A_1 =\sum_{k=0}^\infty A_{1,k} x^k_1$, with
\[A_{1,0} = \left(\begin{matrix} 1 & 2 & 0 & 0 \\ 0 & 0 & 0 & 0 \\ -2 & 0 & 0 &
      0 \\ 0 & 1 & 0 & 0\end{matrix}\right)\quad\text{and}\quad A_{1,1} =
  \left(\begin{matrix} 4 & 9 & 2 & -5 \\ 8 & 9 & 0 & 0 \\ 8 & 6 & 2 & 4 \\ 5 & 6
      & 3 & 3\end{matrix}\right) . \]
The transformation $S = \operatorname{Diag}(x_1,x_1,1,1)$ yields the equivalent
system $[\tilde{A}]$ given by
$$x_1^{p_1+1} \frac{\partial}{\partial x_1} F = \tilde{A_1} F \quad \text{where}
\quad \tilde{A_1} = S^{-1} A_1 S - x_1^p \operatorname{Diag}(1,1,0,0) .$$
As we are interested in the effect of $S$ on the first few terms of $A_1$, we
look into $S^{-1} A_1 S$ which exchanges the upper right and lower left
$2\times2$ blocks of the $A_1$ as exhibited in the following diagram:

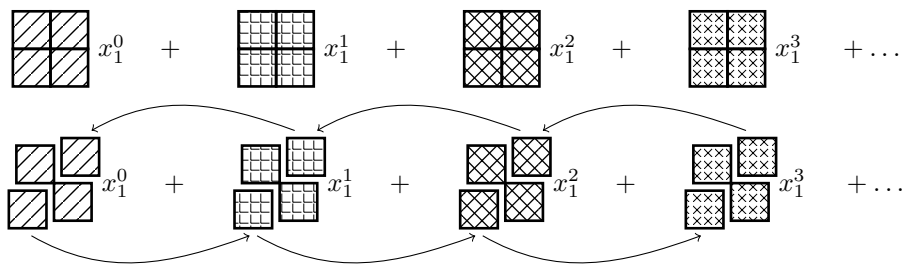
\vspace{0.2cm}
\begin{center}
\begin{tikzpicture}
\draw [line width=1pt,pattern=section]
  (0,0.5) rectangle (0.5,1);
\draw [line width=1pt,pattern=section]
  (0,0) rectangle (0.5,0.5);
\draw [line width=1pt,pattern=section]
  (0.5,0) rectangle (1,0.5) node[right] {$x_1^0\quad\ +$};
\draw [line width=1pt,pattern=section]
  (0.5,0.5) rectangle (1,1);

\draw [line width=1pt,thickness=4pt, pattern=section4]
  (3,0.5) rectangle (3.5,1);
\draw [line width=1pt,thickness=4pt, pattern=section4]
  (3,0) rectangle (3.5,0.5);
\draw [line width=1pt,thickness=4pt, pattern=section4]
  (3.5,0) rectangle (4,0.5) node[right] {$x_1^1\quad\ +$};
\draw [line width=1pt,thickness=4pt, pattern=section4]
  (3.5,0.5) rectangle (4,1);

\draw [line width=1pt,thickness=0.5pt, pattern=section2]
  (6,0.5) rectangle (6.5,1);
\draw [line width=1pt,thickness=0.5pt, pattern=section2]
  (6,0) rectangle (6.5,0.5);
\draw [line width=1pt,thickness=0.5pt, pattern=section2]
  (6.5,0) rectangle (7,0.5) node[right] {$x_1^2\quad\ +$};
\draw [line width=1pt,thickness=0.5pt, pattern=section2]
  (6.5,0.5) rectangle (7,1);

\draw [line width=1pt,thickness=2pt, pattern=section3]
  (9,0.5) rectangle (9.5,1);
\draw [line width=1pt,thickness=2pt, pattern=section3]
  (9,0) rectangle (9.5,0.5);
\draw [line width=1pt,thickness=2pt, pattern=section3]
  (9.5,0) rectangle (10,0.5)
  node[right] {$x_1^3\quad\ +\dots$};
\draw [line width=1pt,thickness=2pt, pattern=section3]
  (9.5,0.5) rectangle (10,1);
\end{tikzpicture}

\begin{tikzpicture}
\draw [line width=1pt,pattern=section]
  (0,0.5) rectangle (0.5,1);
\draw [line width=1pt,pattern=section]
  (-0.1,-0.1) rectangle (0.4,0.4);
\draw [line width=1pt,pattern=section]
  (0.5,0) rectangle (1,0.5) node[right] {$x_1^0\quad\ +$};
\draw [line width=1pt,pattern=section]
  (0.6,0.6) rectangle (1.1,1.1);

\draw [line width=1pt,thickness=4pt, pattern=section4]
  (3,0.5) rectangle (3.5,1);
\draw [line width=1pt,thickness=4pt, pattern=section4]
  (2.9,-0.1) rectangle (3.4,0.4);
\draw [line width=1pt,thickness=4pt, pattern=section4]
  (3.5,0) rectangle (4,0.5) node[right] {$x_1^1\quad\ +$};
\draw [line width=1pt,thickness=4pt, pattern=section4]
  (3.6,0.6) rectangle (4.1,1.1);

\draw [line width=1pt,thickness=0.5pt, pattern=section2]
  (6,0.5) rectangle (6.5,1);
\draw [line width=1pt,thickness=0.5pt, pattern=section2]
  (5.9,-0.1) rectangle (6.4,0.4);
\draw [line width=1pt,thickness=0.5pt, pattern=section2]
  (6.5,0) rectangle (7,0.5) node[right] {$x_1^2\quad\ +$};
\draw [line width=1pt,thickness=0.5pt, pattern=section2]
  (6.6,0.6) rectangle (7.1,1.1);

\draw [line width=1pt,thickness=2pt, pattern=section3]
  (9,0.5) rectangle (9.5,1);
\draw [line width=1pt,thickness=2pt, pattern=section3]
  (8.9,-0.1) rectangle (9.4,0.4);
\draw [line width=1pt,thickness=2pt, pattern=section3]
  (9.5,0) rectangle (10,0.5) node[right] {$x_1^3\quad\ +\dots$};
\draw [line width=1pt,thickness=2pt, pattern=section3]
  (9.6,0.6) rectangle (10.1,1.1);

\path[->] (3.7,1.2) edge [out= 160, in= 30] (1,1.2);
\path[->] (6.7,1.2) edge [out= 160, in= 30] (4,1.2);
\path[->] (9.7,1.2) edge [out= 160, in= 30] (7,1.2);

\path[->] (0.2,-0.2) edge [out= -30, in= -160] (3.1,-0.2);
\path[->] (3.2,-0.2) edge [out= -30, in= -160] (6.1,-0.2);
\path[->] (6.2,-0.2) edge [out= -30, in= -160] (9.1,-0.2);
\end{tikzpicture}

\begin{tikzpicture}
\draw [line width=1pt,pattern=section]
  (0,0.5) rectangle (0.5,1);
\draw [line width=1pt]
  (0,0) rectangle (0.5,0.5) node[pos=0.5] {0};
\draw [line width=1pt,pattern=section]
  (0.5,0) rectangle (1,0.5) node[right] {$x_1^0\quad\ +$};
\draw [line width=1pt,thickness=4pt,pattern=section4]
  (0.5,0.5) rectangle (1,1);

\draw [line width=1pt,thickness=4pt, pattern=section4]
  (3,0.5) rectangle (3.5,1);
\draw [line width=1pt,pattern=section]
  (3,0) rectangle (3.5,0.5);
\draw [line width=1pt,thickness=4pt, pattern=section4]
  (3.5,0) rectangle (4,0.5) node[right] {$x_1^1\quad\ +$};
\draw [line width=1pt,thickness=0.5pt, pattern=section2]
  (3.5,0.5) rectangle (4,1);

\draw [line width=1pt,thickness=0.5pt, pattern=section2]
  (6,0.5) rectangle (6.5,1);
\draw [line width=1pt,thickness=4pt, pattern=section4]
  (6,0) rectangle (6.5,0.5);
\draw [line width=1pt,thickness=0.5pt, pattern=section2]
  (6.5,0) rectangle (7,0.5) node[right] {$x_1^2\quad\ +$};
\draw [line width=1pt,thickness=2pt, pattern=section3]
  (6.5,0.5) rectangle (7,1);

\draw [line width=1pt,thickness=2pt, pattern=section3]
  (9,0.5) rectangle (9.5,1);
\draw [line width=1pt,thickness=0.5pt, pattern=section2]
  (9,0) rectangle (9.5,0.5);
\draw [line width=1pt,thickness=2pt, pattern=section3]
  (9.5,0) rectangle (10,0.5) node[right] {$x_1^3\quad\ +\dots$};
\draw [line width=1pt]
  (9.5,0.5) rectangle (10,1) node[pos=0.5] {$*$};
\end{tikzpicture}
\end{center}
\vspace{0.1cm}

\noindent Consequently, $A_{1,0}$ and $A_{1,1}$ become
  \[\tilde{A}_{1,0} = \left(\begin{matrix} 1 & 2 & 2 & -5 \\ 0 & 0 & 0 & 0 \\ 0
        & 0 & 0 & 0 \\ 0 & 0 & 0 & 0\end{matrix}\right),\ \ \tilde{A}_{1,1} =
    \left(\begin{matrix} 4 & 9 & * & * \\ 8 & 9 & * & * \\ -2 & 0 & 2 & 4 \\ 0 &
        1 & 3 & 3\end{matrix}\right).\]
  Note that the lower left zero entries in $\tilde{A}_{1,0}$ come from
  $A_{1,-1}$, which is a zero matrix. In return, the upper right block of
  $A_{1,0}$ is sent to $A_{1,-1}$. Since it is a zero block, this transformation
  does not introduce denominators. The upper right entries in $\tilde{A}_{1,1}$
  come from $A_{1,2}$. With this transformation, we reduced the rank of
  $A_{1,0}$ from $2$ to $1$.
\end{example}
More generally, let $[A]$ be a multivariate Pfaffian system. The transformation
$F = SG$, where $S$ is a shearing transformation in~$x_i$, yields an equivalent
system with:
\[
\begin{cases} 
\tilde{A}_{i} = S^{-1} A_{i} S - x_{i}^{p_{i}}\operatorname{Diag}({\beta_1}, \dots ,
{\beta_d}), \\
\tilde{A}_{j}  = S^{-1} {A_{j}} S, \quad 1 \leq j \neq i \leq n , \end{cases}\]
where 
\[S^{-1} {A_{j}} S = \left(\begin{matrix} A_{j, 11} & A_{j,12}
      x_{i}^{\beta_2 - \beta_1} & \dots &  A_{j, 1d} x_{i}^{\beta_d - \beta_1} \\[5pt]
      A_{j, 21} x_{i}^{\beta_1 - \beta_2} & A_{j, 22} & \dots & A_{j, 2d}
      x_{i}^{\beta_d- \beta_2} \\[5pt] \vdots & \vdots & \dots & \vdots
      \\[5pt] A_{j, d1} x_{i}^{\beta_1 - \beta_d} & A_{j, d2} x_{i}^{\beta_2
        - \beta_d} & \dots & A_{j, dd} \end{matrix}\right)\text{ for all }1\leq
  j \leq n.\]

The shearing in Example~\ref{ex:shearing} reduced the rank of the leading matrix
coefficient and was compatible with the system, i.e.\ it did not introduce
undesired denominators of $x_{i}$, because of the column reduced form of
$A_{i,0}$. The input system is not always given in such a
form for $A_{i,0}$, and so we investigate in the following subsection how to
achieve it.

\subsection{Column Reduction}
\label{colpfaff}
To enable rank reduction, we alternate between the shearing transformation and
transformations which reduce some columns of a leading  matrix coefficient to
zero. For this we discuss in this section the following problem.

\medskip
\begin{quote}(P) \quad Given a square matrix
  $A=[v_{1},\dots,v_{d}]\in\operatorname{Mat}_{d\times d}({\rm R})$ (where
  $v_{i}$ denotes the $i$th column) of rank $r<d$ when considered as an element
  of $\operatorname{Mat}_{d\times d}({\rm K})$, find a
  transformation $T\in GL_d({\rm R})$ such that the last~$d-r$ columns of
  $TAT^{-1}$ are zero. \end{quote}

\medskip
Before considering the algorithmic aspects,
we first discuss the existence of such a transformation. As the next example
shows, the desired transformation does not necessarily exist for any matrix $A$.
\begin{example}
The matrix
\[
\left(\begin{matrix}
0          & x_1 & x_2\\
0   & 0 & 0\\
0 & 0 & 0
\end{matrix}\right)
\]
is obviously of rank 1 over $\rm K$. There is, however, no
transformation~$T\in GL_3({\rm R})$ such that $TAT^{-1}$ contains only one
non-zero column.
\end{example}
Consider the finitely generated ${\rm R}^d$-submodule
$M:=\langle v_{1},\dots,v_{d}\rangle$. We call it the column module of $A$. In
order to construct a suitable transformation for bivariate Pfaffian systems (for
which the leading matrix coefficients are univariate) the authors
of~\cite{key101} use the fact that $\set C[[x_1]]$ and $\set C[[x_2]]$ are
principal ideal domains and hence that every finitely generated submodule of a
free module over this ring is free. We generalize this for the multivariate case
by showing in Corollary~\ref{cor:basis} that the freeness of the column module
$M$ is a necessary and sufficient condition for the existence of a
transformation that meets our requirements. This is a direct consequence of
Nakayama's Lemma for local rings.

\begin{theorem} \cite[Theorem 2.3, pp 8]{key900}
  \label{theo:naka}
  Let $\rm R$ be a local ring, $\mathcal{M}$ its maximal ideal and let $M$ be a
  finitely generated $\rm R$-module. Then $v_1,\dots,v_r\in M$ form a minimal
  set of generators for $M$ if and only if their images
  $\bar{v}_1,\dots,\bar{v}_r$ under the canonical homomorphism
  $M\rightarrow M/\mathcal{M}M$ form a basis of the vector space
  $M/\mathcal{M}M$ over the field ${\rm R}/\mathcal{M}$.
\end{theorem}

The central consequence of Theorem~\ref{theo:naka} for us is that if $M$ is
free, a module basis of $M$ can be chosen among the columns of $A$. We adapt the
theorem to our situation to show that we can bring $A$ into a column-reduced
form if and only if its column module is free.

\begin{corollary}
  \label{cor:basis}
  Let $A\in\operatorname{Mat}_{d\times d}({\rm R})$ be of rank $r$ over $\rm K$
  and let $M$ be the module generated by the columns of $A$. If $M$ is free,
  then there exists a subset $B$ of the columns in $A$ with $r$ elements such
  that $B$ is a module basis of $M$. Furthermore, $B$ is also a $\rm K$-vector
  space basis of the column space of $A$.
\end{corollary}

\begin{proof}
  By Theorem~\ref{theo:naka} we can find a basis $B=\{b_1,\dots,b_k\}$ of $M$
  among the columns of~$A$. By definition, the $b_i$ are linearly independent
  over $\rm R$, so they are also linearly independent over
  $\rm K$ (otherwise, multiplying a linear relation in
  $\rm K$ with a common denominator yields a relation in
  $\rm R$). Since $B$ is a basis of the column module, it also contains a
  generating set of the $\rm K$-vector space generated by
  the columns of $A$.
\end{proof}

In theory, Corollary~\ref{cor:basis} would allow the computation of a unimodular
column reduction transformation as required in $(P)$ simply via Gaussian
elimination. Assume we are given a matrix $A$ and already know a subset
$B=(b_1,\dots,b_r)$ of the columns of $A$ which forms a basis of the column
module.  Let $v$ be a column vector of $A$ which is not in~$B$. Then, since $B$
is a vector space basis, there exist $c_1,\dots,c_r\in {\rm K}$ such that
\[c_1b_1+\dots +c_rb_r=v.\]
By assumption, $B$ is also a module basis, so there also exist
$d_1,\dots,d_r\in {\rm R}$ with
\[d_1b_1+\dots +d_rb_r=v.\]
The $b_i$ are linearly independent, and therefore the cofactors of $v$ with
respect to $B$ are unique. It follows that $c_i=d_i$ for all $1\leq i \leq
r$.

The main algorithmic difficulty stems from the fact that not all formal power
series admit a finite representation, and even if the initial system is given in
a finite form, the splitting transformation as in Theorem~\ref{blockpfaff} does
not preserve finiteness. In particular, we face two main problems when working
with truncated power series:
\begin{description}
\item [$(P1)$] Detecting the correct rank and the linear independent columns of
  $A$
\item [$(P2)$] If we know the independent columns, a column reduction
  transformation computed after truncation is not uniquely determined.
\end{description}
These computational problems arise for general multivariate and for bivariate
systems, but were not addressed in previous algorithmic works on this
topic~\cite{key101,key5,key73}. Before we propose our resolution, we illustrate
both problems in the following example:

\begin{example}
\label{exm:trunc}
Consider the matrix 
\[\left(\begin{matrix}
x & 0 & x^2 & x^2+x\\
0 & x & x & x\\
1 & 0 &0 & 1
\end{matrix}\right).
\]
Here, the first three columns $v_1,v_2,v_3$ are linearly
independent and generate the column module. A linear combination of the fourth
column $v_4$ is given by
\[1\cdot v_1 + 0\cdot v_2+1\cdot v_3 = v_4.\]
When truncating at order 1, the system is given as 
\[\left(\begin{matrix}
x & 0 & 0 & x\\
0 & x & x & x\\
1 & 0 &0 & 1
\end{matrix}\right)
\]
The original rank cannot be determined from the truncated matrix. Furthermore,
even if we know that $v_1,v_2,v_3$ are linearly independent, there are several
linear combinations of the fourth column after truncation:
\[1\cdot v_1 + 0\cdot v_2+1\cdot v_3 = v_4.\]
\[1\cdot v_1 + 1\cdot v_2+0\cdot v_3 = v_4.\]
The cofactors of the second linear combination are not the truncated cofactors
of the first. It can not be extended with higher order terms to a suitable linear
combination over the formal power series ring without truncation.
\end{example}

We can solve both $(P1)$ and $(P2)$ with the help of minors of the original
system. Let $r$ be the rank of $A$. Then there exists a nonzero $r\times r$
submatrix $B$ of $A$ whose determinant is nonzero. Let $k$ be the order of the
determinant. If we take the truncated system
$\tilde{A} = A \operatorname{rem} x^{k+1}$, the same submatrix $\tilde{B}$ in
$\tilde{A}$ will have a non-zero determinant modulo $x^{k+1}$ and we can
therefore identify in~$\tilde{A}$ which columns in $A$ are linearly
independent. This resolves $(P1)$ as long as the truncation order $k$ is chosen
big enough.

Next assume that for instance the first $r$ columns of $A$ are linearly
independent, i.e.\ we can choose $B$ such that its columns correspond to
$v_1,\dots,v_r$. Let $k$ be as above, $\ell$ be a positive integer and let $v$
be a column vector that is linearly dependent on the columns of $B$. Then there
exist $c_1,\dots,c_r\in\set C[[x_1,\dots,x_n]]$ such that
\[B\cdot(c_1,\dots,c_r)=v.\]
By Cramer's rule, we know that the $c_i$ are given by
\begin{equation}
\label{cramer}
c_i=\frac{\det(B_i)}{\det(B)},
\end{equation}
where $B_i$ is the matrix obtained by replacing the $i^{th}$ column of $B$ by
$v$. Rewriting Equation~\eqref{cramer} gives
\begin{equation}
\label{cramer2}
\det(B)c_i-\det(B_i)=0,
\end{equation}
and this equation allows the computation of $c_i$ by coefficient comparison. In
particular, we are guaranteed to obtain the correct $c_i$ up to order $\ell$ if
in~\eqref{cramer2} we replace $B$ by $\tilde{B}$, its truncation at order
$\ell+k+1$, and $B_i$ by $\tilde{B}_i$, the truncation of $B_i$ at order
$\ell+k+1$. This resolves $(P2)$. 

This approach is based on the fact that there is a truncation order $k$ such
that we can find a submatrix of maximal dimension with non-zero determinant. We
have to remark, however, that by the nature of formal power series, it is in
general not possible to tell a priori if a given truncation is high
enough.  Furthermore, we emphasize that it is in general not
  possible to draw a conclusion about the freeness of the column module from the
  integral relations among the truncated column vectors, since any linear
  combination of the form $c_1v_1+\dots+c_{d-1}v_{d-1}-v_d=0 \mod x^k$ can
  require a non-unit cofactor $c_d$ for higher truncation orders. However, if no
  integral relations can be found with the above method, also the column module
  without truncation cannot be free. Both observations lead to the following
  practical approach. The full algorithm is carried out with a given truncation
  order. If the output is correct (compared to the invariant exponential part which can be obtained by Theorem~\ref{exponentialpfaff}), we are done. If not, we increase the
  truncation order until we get a correct output or arrive at a point where no
  integral relations can be found anymore.  This procedure necessarily
terminates, since there exists a suitable truncation order.  

One should note
that not every $\rm K$-vector space basis of the column space of $A$ is also a
module basis. So, in the worst case, $\binom{d}{r}$ submatrices have to be
tested to obtain a module basis.

% Subsequently, we will refer to the following conditions whenever
% necessary:

% \begin{quote}
%   We say that $(\mathcal{C})$ (respectively ($\mathcal{R}$)) is satisfied if the column
%   (respectively row) module of the matrix under consideration is free.
% \end{quote}
% Evidently, $(\mathcal{C})$ (respectively ($\mathcal{R}$)) is always satisfied in the
% case $n=2$ since ${\rm R}_{\bar{x}_i}$ would be a principal ideal domain.

\subsection{Converse of Theorem~\ref{moserpfaff}}

\label{proofpfaff}
We consider again a multivariate system $[A]$ as in~\eqref{eq:sys}. We fix
$i \in \{ 1, \dots, n \}$ and investigate the rank reduction of its ${i}^{th}$
component given by
\begin{equation}
\label{first}
x_{i}^{p_{i}+1} \pder{F}{x_i} = A_{i}F = (A_{i,0} + A_{i,1} x_{i} + A_{i,2}
x^2_{i} + A_{i,3} x^3_{i} + \dots )F ,
\end{equation}
where the matrices $A_{i,j}$ have their entries in ${\rm R}_{\bar{x}_i}$ and the
algebraic rank of $A_{i,0}$ is denoted by $r$. We recall that we defined
$\bar{x}_{i}:= (x_1, \dots, x_{i-1},x_{i+1}, \dots, x_n)$ and ${\rm R}_{\bar{x}_i}
:= \mathbb{C}[[\bar{x}_{i} ]]$.

The establishment of the converse of Theorem~\ref{moserpfaff} for the reduction in $x_i$ follows
essentially the steps of that of the bivariate case which was given
in~\cite{key101}. The construction requires successive application of
transformations in $GL_d(R_{\bar{x}_i})$ and shearing transformations in
$x_i$. We remark that when applying a transformation
$T \in GL_d({\rm R}_{\bar{x}_{i}})$ on the $i^{th}$ component,~\eqref{eq:gauge}
reduces to
$$ \tilde{A}_{i}= T^{-1}  A_{i} T .$$
Given~\eqref{first}, if the column module of ${A}_{i,0}$ is free, then one can
compute a transformation $U_1 \in GL_d({\rm R}_{\bar{x}_{i}})$ such that
  $$U_1^{-1}{A}_{i,0} U_1  =  \left(\begin{matrix} B^{11} & O
      \\ B^{21} & O \end{matrix}\right)$$
  has rank $r$, entries in ${\rm R}_{\bar{x}_i}$ and with diagonal blocks of
  sizes $r\times r$ and $(d-r)\times(d-r)$ respectively. Let~$v$ be the rank of
  $B^{11}$.  If also the column module of $B^{11}$ is free, then one
  can compute a transformation $U_2 \in GL_r({\rm R}_{\bar{x}_i})$
  such that
  $$U_2^{-1}B^{11} U_2  =  \left(\begin{matrix} E^{11} & O  \\
       E^{21} & O \end{matrix}\right)$$
  has rank $v$, entries in ${\rm R}_{\bar{x}_i}$ and with diagonal blocks of
  sizes $v\times v$ and $(r-v)\times(r-v)$ respectively. We set
  $U : = \operatorname{Diag} (U_2,I_{d-r})\cdot U_1$.  Then the leading
  coefficient $\tilde{A}_{i,0}$ of the equivalent system $U[A_{i}]$ has the
  following form:
\begin{equation} \label{gaussformPfaffian} 
\tilde{A}_{i,0}
  \;=\; \left(\begin{matrix}\tilde{A}_{i,0}^{11} & O& O\\[5pt]
      \tilde{A}_{i,0}^{21} & O_{r-v}  & O \\[5pt] \tilde{A}_{i,0}^{31} &
      \tilde{A}_{i,0}^{32} & O_{d-r} \end{matrix}\right)\end{equation}
with diagonal blocks of sizes $v\times v$, $(r-v)\times(r-v)$  and
$(d-r)\times(d-r)$ respectively for some $0\leq v < r$ and where 
\[
  \left(\begin{matrix} \tilde{A}_{i,0}^{11} \\[5pt]
      \tilde{A}_{i,0}^{21} \end{matrix}\right) \quad \text{and} \quad
  \left(\begin{matrix}\tilde{A}_{i,0}^{11} & O \\[5pt] \tilde{A}_{i,0}^{21} & O
      \\[5pt] \tilde{A}_{i,0}^{31} & \tilde{A}_{i,0}^{32} \end{matrix}\right)\]
are $r \times v$ and $d \times r$ matrices of full column ranks $v$ and $r$
respectively. Clearly, $U$ is compatible with
system~$A$ since it is unimodular.\\
From now on, we assume that the leading coefficient $A_{i,0}$ of~\eqref{first}
is in form~\eqref{gaussformPfaffian}. In particular, we require the
column module of $A_{i,0}$ and the column module of $B^{11}$ as given above to
be free. We then partition $A_{i,1}$ in accordance with $A_{i,0}$ and set
\begin{equation}
\label{glambdaformPfaffian} 
G_{A_{i}} (\lambda) := \left(\begin{matrix} A_{i,0}^{11} & O & A_{i,1}^{13}
    \\[5pt] A_{i,0}^{21} & O & A_{i,1}^{23} \\[5pt] A_{i,0}^{31} & A_{i,0}^{32} &
    A_{i,1}^{33}+ \lambda I_{d-r}\end{matrix}\right) .\end{equation}
The polynomial $\det(G_{A_{i}} (\lambda))$ vanishes identically in
$\lambda$ if and only if $\theta_{i} (\lambda)$ given by~\eqref{eq:theta}
does. In fact, let
$D(x_i)= \operatorname{Diag}(x_i I_{r}, I_{d-r})$. Then we can write
${x_i}^{-1} A_i(x) = N(x) D^{-1}(x_i)$ where $N(x) \in {\rm R}^{d \times d}$, and
set $D_0 = D(x_i=0)$, $N_0= N(x_i=0)$. Then we have
\begin{eqnarray*} \det (G_{A_i}(\lambda)) & =& \det (N_0 + \lambda D_0) = \det
(N + \lambda D)|_{x_i=0} \\ & = & (\det(\frac{A}{x_i} + \lambda I_d) \det
(D))|_{x_i=0} \\ &=& (\det (\frac{A_{i,0}}{x_i} + A_{i,1} + \lambda I_d)
{x_i}^{r} )|_{x_i=0} = \theta_{i} (\lambda).  \end{eqnarray*}
Moreover, $G_{A_{i}} (\lambda)$ has an additional important application within the
construction of a desired transformation as we show in the following proposition:
\begin{proposition}
\label{gauss3Pfaffian}
Suppose that $m(A_{i}) >1$ and $\det(G_{A_{i}} (\lambda))$ is identical to
zero. If the row module of $G_{A_{i}} (\lambda =0)$ is free then there exists a
transformation $Q({\bar{x}_i})$ in $GL_d({\rm R}_{\bar{x}_i})$ with
$\det(Q) = \pm 1$, compatible with system~$A$, such that the matrix
$G_{\tilde{A}_{i}} (\lambda))$ has the form
\begin{equation} 
\label{particularform3Pfaffian} 
G_{\tilde{A}_{i}} (\lambda) = \left(\begin{matrix} A_{i,0}^{11} & O & U_1& U_2
\\[5pt] A_{i,0}^{21} & O & U_3 & U_4 \\[5pt] V_1 & V_2 & W_1 + \lambda I_{d- r
-\varrho} & W_2 \\[5pt] M_1 & O & M_3 & W_3 + \lambda
I_\varrho \end{matrix}\right) ,\end{equation} where $0 \leq \varrho \leq d-r$,
and
\begin{alignat}2
\label{particularconditionbPfaffian} 
  & \operatorname{rank}\left(\begin{matrix} A_{i,0}^{11} &U_1\\[5pt] A_{i,0}^{21} &
      U_3 \\[5pt] M_1& M_3 \end{matrix}\right) &=\; &
\operatorname{rank}\left(\begin{matrix} A_{i,0}^{11} & U_1\\[5pt] A_{i,0}^{21} &
U_3 \end{matrix}\right), \\[10pt] \label{particularconditionaPfaffian}
&\operatorname{rank} \left(\begin{matrix} A_{i,0}^{11} & U_1\\[5pt] A_{i,0}^{21} &
U_3 \end{matrix}\right) &<\; &r. 
\end{alignat}
\end{proposition}
\begin{proof}
  If the row module of $G_{A_{i}} (\lambda =0)$ is free then the transformation $Q ({\bar{x}_i})$ can be constructed as in the proof
  of~\cite[Proposition 3]{key101} for the bivariate case.
\end{proof}
\noindent We remark that in the particular case of
$v=0$,~\eqref{gaussformPfaffian} is given by
$$
\tilde{A}_{i,0} (\bar{x}_{i} ) = \left(\begin{matrix}
    O_{r}  & O \\[5pt]
    \tilde{A}_{i,0}^{32} & O_{d-r} \end{matrix}\right) , \quad \text{with} \quad
\operatorname{rank}(\tilde{A}_{i,0}^{32}) = r .$$
Consequently, it can be easily verified that~\eqref{particularform3Pfaffian} is
given by
$$
G_{\tilde{A}_{i}} (\lambda) = \left(\begin{matrix} O_r & U_3 \\[5pt] V_2 & W_1 +
    \lambda I_{d- r} \end{matrix}\right) , \quad \text{and} \quad \varrho = 0
. $$
\begin{proposition}
\label{shearingPfaffian}
If $m(A_{i}) >1$ and $\det(G_{A_{i}} (\lambda)) \equiv 0$ is as
in~\eqref{particularform3Pfaffian} with
conditions~\eqref{particularconditionbPfaffian}
and~\eqref{particularconditionaPfaffian} satisfied, then the component $A_{i}$
of $A$ in~\eqref{eq:sys} is reducible and reduction can be carried out with the
shearing $F = S(x_i)\;G$ where
 $$\begin{cases} S(x_i)=\operatorname{Diag}(x_i I_r , I_{d-r-\varrho}, x_i
   I_\varrho) \quad \text{if } \varrho \neq 0 \\[5pt]
   S(x_i)=\operatorname{Diag}(x_i I_r , I_{d-r}) \hspace{37px}
   \text{otherwise.} \end{cases} $$
 Furthermore, this shearing is compatible with system $A$.
\end{proposition}
\begin{proof}
  Given system~\eqref{eq:sys}. For any $j \in \{1, \dots, n\}$ we partition
  $A_{j}$ according to~\eqref{particularform3Pfaffian}
$$A_{j} = \left(\begin{matrix} A_{j}^{11}  & A_{j}^{12} & A_{j}^{13}&
    A_{j}^{14}\\[5pt] A_{j}^{21} & A_{j}^{22} & A_{j}^{23} & A_{j}^{24}\\[5pt]
    A_{j}^{31} & A_{j}^{32} & A_{j}^{33} & A_{j}^{34}\\[5pt] A_{j}^{41} &
    A_{j}^{42} & A_{j}^{43} & A_{j}^{44}\end{matrix}\right), \quad 1 \leq j \leq
n, $$
where $A_{j}^{11} , A_{j}^{22} , A_{j}^{33} , A_{j}^{44}$ are square matrices of
dimensions $v, r-v, d-r-\varrho,$ and $\varrho$ respectively.  It is easy to
verify that the equivalent system $S[A] \equiv {\tilde{A}}$ given
by~\eqref{eq:equiv} admits the form
\begin{alignat*}2
  & \tilde{A}_{i} &= &\left(\begin{matrix}\phantom{x_i} A_{i}^{11}
      & \phantom{x_i} A_{i}^{12} & x_{i}^{-1} A_{i}^{13} & \phantom{x_i}
      A_{i}^{14} \\[5pt] \phantom{x_i} A_{i}^{21} & \phantom{x_i} A_{i}^{22}&
      x_{i}^{-1} A_{i}^{23} & \phantom{x_i}A_{i}^{24} \\[5pt] x_{i} A_{i}^{31} &
      x_{i} A_{i}^{32}& \phantom{x_i^{-1}} A_{i}^{33}& x_{i} A_{i}^{34} \\[5pt]
      \phantom{x_i} A_{i}^{41} & \phantom{x_i} A_{i}^{42} & x_{i}^{-1}
      A_{i}^{43} & \phantom{x_i} A_{i}^{44} \end{matrix}\right) - x_{i}^{p_{i}}
  \operatorname{Diag} (I_{r}, O_{d-r-\varrho}, I_\varrho) \\[10pt] 
  & \tilde{A}_{j} &=& \left(\begin{matrix} \phantom{x_j} A_{j}^{11} & \phantom{x_j} A_{j}^{12} &
      x_{j}^{-1} A_{j}^{13} & \phantom{x_j} A_{j}^{14} \\[5pt] \phantom{x_j}
      A_{j}^{21} & \phantom{x_j} A_{j}^{22} & x_{i}^{-1} A_{j}^{23} &
      \phantom{x_j} A_{j}^{24} \\[5pt] x_{j} A_{j}^{31}
      & x_{j} A_{j}^{32} & \phantom{x_j^{-1}}A_{j}^{33} & x_{j} A_{j}^{34} \\[5pt]
      \phantom{x_j} A_{j}^{41} & \phantom{x_j} A_{j}^{42} & x_{j}^{-1}A_{j}^{43}
      & \phantom{x_j} A_{j}^{44} \end{matrix}\right) , \quad 1 \leq j \neq i
  \leq n.
\end{alignat*} 
Hence, the leading matrix coefficient of the equivalent
$i^{th}$-component is given by
$$ \tilde{A}_{i,0} (\bar{x}_{i}) = \left(\begin{matrix} A_{i,0}^{11}
    & O & U_1 & O \\[5pt] A_{i,0}^{21} & O & U_3 & O \\[5pt] O& O &O& O \\[5pt]
    M_1 & O & M_3 & O \end{matrix}\right) $$
where $\operatorname{rank}(\tilde{A}_{i,0}) < r$ since~\eqref{particularconditionbPfaffian}
and~\eqref{particularconditionaPfaffian} are satisfied.
\goodbreak

\noindent It remains to prove the compatibility of $S(x_i)$ with the
system~\eqref{eq:sys}, in particular, that the normal crossings are
preserved. It suffices to prove that the entries of
${A}_{j} , 1 \leq j \neq i \leq n, $ which will be multiplied by $x_{i}^{-1}$
upon applying $S(x_i)$, namely, the entries of $A_{j}^{13}, A_{j}^{23},$ and
$A_{j}^{43}$ are zero matrices modulo $x_i$ otherwise poles in $x_{i}$ will be
introduced. This can be restated as requiring $A_{j}^{13}({x}_{i} = 0),$
$A_{j}^{23}({x}_{i} = 0),$ and $A_{j}^{43}({x}_{i} = 0)$ to be zero
submatrices. This requirement is always satisfied due to the integrability
condition and the resulting equality, obtained by setting $x_i=0$, which we
restate here
 \begin{equation}
 \label{eq:cond0}
 x_j^{p_j+1}\pder{A_{i,0}}{x_j}= A_{j}(x_i=0)\;A_{i,0} - A_{i,0} A_{j}(x_i=0) ,
\quad 1 \leq j \neq i \leq n .
\end{equation}
This equality induces a structure of $A_{j}(x_i=0)$ which depends on that of
$A_{i,0}$. Since $G_{A_{i}} (\lambda)$ is as in~\eqref{particularform3Pfaffian},
then, before applying the shearing transformation,
$A_{i,0} ({\bar{x}}_{i})$ has the following form~\eqref{form0} and
$A_{j} (x_i = 0)$ can be partitioned accordingly. So we have for
$1 \leq j \neq i \leq n$ 
\begin{eqnarray} \label{form0} A_{i,0}({\bar{x}}_{i}) &=&\left(\begin{matrix}
A_{i,0}^{11} & O & O & O \\[5pt] A_{i,0}^{21} & O_{(r-v)(r-v)} & O & O \\[5pt]
V_1 & V_2 & O_{(d-r-\varrho)(d-r-\varrho)} & O \\[5pt] M_1 & O& O & O_{\varrho
\varrho} \end{matrix}\right) ,\\[10pt]
 \label{form1} A_{j}(x_i=0) &=& \left(\begin{matrix} A_{j}^{11}(x_i=0) &
A_{j}^{12}(x_i=0)& A_{j}^{13}(x_i=0)& A_{(j)}^{14}(x_i=0) \\[5pt]
A_{j}^{21}(x_i=0) & A_{j}^{22}(x_i=0) & A_{j}^{23}(x_i=0) &
A_{j}^{24}(x_i=0)\\[5pt] A_{j}^{31}(x_i=0) & A_{j}^{32}(x_i=0) &
A_{j}^{33}(x_i=0)& A_{j}^{34}(x_i=0)\\[5pt] A_{j}^{41}(x_i=0) &
A_{j}^{42}(x_i=0) & A_{j}^{43}(x_i=0) &
A_{j}^{44}(x_i=0) \end{matrix}\right). \end{eqnarray} Inserting~\eqref{form0}
and~\eqref{form1} in~\eqref{eq:cond0}, one can obtain the desired results by
equating the entries of~\eqref{eq:cond0}. More explicitly, upon investigating
the entries in (Column 3), (Rows 1 and 2, Column 2), and (Row 4, Column 2), we
observe the following respectively:
\begin{itemize}
\item We have that 
\[\left(\begin{matrix} A_{i,0}^{11} & O \\[5pt] A_{i,0}^{21} & O \\[5pt]
V_1 & V_2 \\[5pt] M_1 & O \end{matrix}\right) \cdot \left(\begin{matrix}
A_{j}^{13}(x_i=0) \\[5pt] A_{j}^{23}(x_i=0) \end{matrix}\right) = O_{n,
n-r-\varrho}.\] The former matrix is of full rank $r$ by construction thus $
\smash{\left(\begin{matrix} A_{j}^{13}(x_i=0) \\[5pt]
A_{j}^{23}(x_i=0) \end{matrix}\right)}$ is a zero matrix.
\item We also get
\[\smash[t]{\left(\begin{matrix} A_{i,0}^{11} \\ A_{i,0}^{21} \end{matrix}\right)}
\cdot A_{j}^{12}(x_i=0) = O_{r, r-v}.\] The former is of full rank $v$ by
construction thus $A_{j}^{12}(x_i=0)$ is a zero matrix.
\item Finally, $A_{j}^{43}(x_i=0)\cdot V_2 - M_1 \cdot A_{j}^{12}(x_i=0)
=O_{\varrho, (r-v)}$. Since $A_{j}^{12}(x_i=0)$ is null and $V_2$ is of full
column rank $r-v$ by construction then $A_{j}^{43}(x_i=0)$ is a zero matrix as
well.
\end{itemize} This completes the proof.\end{proof}

We can thus establish the following theorem:

\begin{theorem}
\label{moserpfaff2}
Consider a Pfaffian system~$[A]$ and suppose that $m(A_i) > 1$ for some index $i \in \{1, \dots, n\}$. If $\theta_i (\lambda)$ given by~\eqref{eq:theta} vanishes for some index $i \in \{ 1, \dots, n\}$, then under the conditions required to attain~\eqref{gaussformPfaffian} and Proposition~\ref{gauss3Pfaffian}, we can
construct a compatible transformation $T \in GL_d({\rm R}_L)$ which reduces
$m(A_i)$ (and consequently $m(A)$). In this case, $T$ can be chosen to be a
product of transformations in $GL_d({\rm R}_{\bar{x}_i})$ and polynomial
transformations of the form
$\operatorname{Diag}(x_i^{\beta_1}, \dots, x_i^{\beta_d})$ where
$\beta_1, \dots, \beta_d$ are non-negative integers.
\end{theorem}

\begin{proof}
Under the required conditions, we can assume that $A_{i,0} $ has the
  form~\eqref{gaussformPfaffian}. Let $G_{A_{i}} (\lambda)$ be given by~\eqref{glambdaformPfaffian}. Then
  $\det(G_{A_{i}} (\lambda))$ vanishes identically in $\lambda$ if and only
  if $\theta_i(\lambda)$ does. Then the system $S[Q[A]]$ where
  $S, Q$ are as in Propositions~\ref{gauss3Pfaffian} and~\ref{shearingPfaffian}
  respectively, has the desired property.
\end{proof}
For a given index $i\in\{1,\dots,n\}$, the algebraic rank of the leading matrix coefficient can be decreased as long as $\theta_{i} (\lambda) $ vanishes
identically in $\lambda$. In case the leading matrix coefficient eventually
reduces to a zero matrix, the Poincar\'e rank drops at least by one. This
process can be repeated until the Moser rank of system $[A]$ equals to its Moser
invariant. Due to the compatibility of $T$ in Theorem~\ref{moserpfaff}, rank
reduction can be applied to any of the components of $[A]$ without altering the
Moser rank of the others. Hence, by Corollary~\ref{prank}, the true Poincar\'e
rank of system~$[A]$ can be attained by a successive application of the rank
reduction to each of its components.

Finally, we remark that the conditions of Theorem~\ref{moserpfaff2} are always satisfied in the bivariate case $n=2$ of arbitrary dimension.

\subsection{Examples}
\begin{example}
Consider the completely integrable Pfaffian system with normal crossings given by 
  \begin{equation*}
    \begin{cases}
      x_1^{2} \pder{F}{x_1} = A_1 F = \left(\begin{matrix}
          (x_1 x_2 x_3 +1)(x_1-1)& x_3 (x_1 -1) \\
           x_1 x_2 (1 -2x_1 + x_1 x_2 x_3 - x_1^2x_2x_3)& x_1 x_2 x_3 (1-x_1)
        \end{matrix}\right) F, \\
      x_2^3 \pder{F}{x_2} = A_2 F= \left(\begin{matrix}
          (2+3 x_2) (x_1 x_2 x_3 +1) &  x_3 (2+3x_2)\\
          -x_1x_2(3x_1x_2^2x_3 + 2x_1 x_2 x_3 + x_2^2 + 3 x_2 +2) & -x_1 x_2 x_3 (2+3x_2)
        \end{matrix}\right) F,\\
           x_3 \pder{F}{x_3} = A_3 F= \left(\begin{matrix}
          1 & 0 \\
          -x_1 x_2 & 0
        \end{matrix}\right) F .
    \end{cases}
  \end{equation*}
A fundamental matrix of formal solutions is given by:
\begin{equation} \label{exm3:sol} \Phi(x_1, x_2, x_3) x_1^{C_1}x_2^{C_2} x_3^{C_3}
    e^{q_1(x_1^{- 1/s_1})} e^{q_2(x_2^{- 1/s_2})} e^{q_3(x_3^{- 1/s_3})}. \end{equation} 
    
If we only seek to compute the exponential parts $q_1, q_2, q_3$ in~\eqref{exm3:sol}, then from the associated ODS, we compute:
$$ \begin{cases} s_1 =1 \; \text{and} \; q_1(x_1) = \frac{1}{x_1},\\ s_2 =1 \; \text{and} \; q_2(x_2) = \frac{-1 - 3x_2}{x_2^2} , \\ 
s_3 =1 \; \text{and} \; q_3(x_3) =0. \end{cases}$$

Furthermore, if we wish to compute a fundamental matrix of formal solutions,
then following the steps of our formal reduction algorithm, we look at the
leading coefficients of each of the three components of the given system. If one
of these coefficients has two distinct eigenvalues, the we can apply the
splitting lemma (Theorem~\ref{blockpfaff}). Indeed, since $A_{1,0}(x_2, x_3)$ has this property, we can
compute such a transformation $ F = T G$ up to any order. In particular, up to
order $10$, we compute:
$$ T = \begin{pmatrix} 1 & x_1^3 x_2^3 x_3^4 - x_1^2 x_2^2 x_3^3 + x_1 x_2 x_3^2 - x_3 \\ -x_1 x_2 & 1 \end{pmatrix} ,$$
which yields the following diagonalized system (up to order $10$):

\begin{equation*}
    \begin{cases}
      x_1^{2} \pder{G}{x_1} = A_1 G = \left(\begin{matrix}
          x_1 -1 & 0 \\
          0 & x_1^2 x_2 x_3 f(x)
        \end{matrix}\right) G, \\
      x_2^3 \pder{G}{x_2} = A_2 G= \left(\begin{matrix}
          (2+3 x_2)  &  0\\
          0 & x_1 x_2^3 x_3 f(x)
        \end{matrix}\right) G,\\
           x_3 \pder{G}{x_3} = A_3 G= \left(\begin{matrix}
          1 & 0 \\
        0 & x_1 x_2 x_3 f(x)
        \end{matrix}\right) G,
    \end{cases}
  \end{equation*}
with $f(x):=(x_1^2 x_2^2 x_3^2 - x_1 x_2 x_3 +1)$.
Hence, the system can be uncoupled into two subsystems of linear scalar
equations and integrated to construct $G$, and consequently,~$F$.
\end{example}

An example of the reduction process for a system in two variables was given in
Example~\ref{ex:sim}. However, examples of dimensions two and three do not cover
the richness of the techniques presented. So, to illustrate the full process, we
treat an example of dimension six. Due to the size of the system and the number
of necessary computation steps, we are not able to include it directly in this
paper. It is available in several formats at \vspace{0.1cm}
\begin{center}
\url{http://www.mjaroschek.com/pfaffian/}
\end{center}
\section{Formal Reduction Algorithm}
\label{mainpfaff}
\subsection{The Algorithm In Pseudo Code}
We now give the full algorithm in pseudo-code and we refer to more
detailed descriptions within the article whenever necessary.

\begin{remark}
  Throughout the article, we adopted the field of complex numbers $\set C$ as
  the base field for the simplicity of the presentation. However, any computable
  commutative field $K$ with $\mathbb{Q} \subseteq K \subseteq \mathbb{C}$ can
  be considered instead. In this case, the restrictions on the extensions of the
  base field discussed in~\cite{key24} apply
  as well and are taken into consideration within our \textsc{Maple}
  implementation.
\end{remark}

Given system $[A]$, we discuss the eigenvalues of the leading matrix coefficients 
$A_{i,0}$, $i \in \{ 1, \dots, n\}$, of its $n$ components. If for all of these
components uncoupling is unattainable, then we fix $i \in \{ 1, \dots, n\}$ and
proceed to compute the exponential order $\omega(A_i)$ from the associated ODS.
Suppose that $\omega(A_i)= \frac{{\ell} }{m}$ with ${\ell} ,m$ coprime positive
integers. One can then set $t = {x_i}^{1/m}$ (re-adjustment of the independent
variable), and perform again rank reduction to get an equivalent system whose
$i^{th}$ component has Poincar\'e rank equal to ${\ell}$ and leading matrix
coefficient with at least $m$ distinct eigenvalues. Consequently,
block-diagonalization can be re-applied to uncouple the $i^{th}$-component.  By
Section~\ref{diagpfaff}, this uncoupling results in an uncoupling for the
whole system. As mentioned before, this procedure can be repeated until we
attain either a scalar system, i.e.\ a system whose $n$ components are scalar
equations, or a system whose Poincar\'e rank is given by $(0, \dots,0)$. The
former is trivial and effective algorithms are given for the latter
in~\cite[Chapter 3]{key73}. \vspace{1cm}

\bgroup
\def\arraystretch{1.0}
\begin{center}
  \begin{tabular}{p{11cm}l}
    \hline\\[-2.5ex]
  \textbf{Algorithm 1: fmfs\_pfaff(}$p,A$\textbf{)}\\
    \hline\\[-2.5ex]
    \begin{tabular}{lp{9cm}}
      \textbf{Input:} & $p=(p_1,\dots,p_n), A(x)=(A_1,\dots,A_n)$ of~\eqref{eq:sys}.\\
      \textbf{Output:} & A fundamental matrix of formal solutions~\eqref{eq:sol}.\\
    \end{tabular} & \\
    \hline\\[-2.5ex]
    \begin{tabular}{ll}
\quad ${\{C_i\}}_{1 \leq i \leq n} \gets  O_n$; ${\{Q_i\}}_{1 \leq i \leq n} \gets  O_d$;  $\Phi \gets I_d$\\
\quad \textbf{WHILE} $d \neq 1$ or $p_{i} >0$ for some $i \in \{ 1, \dots , n\}$ \textbf{DO} \\
 \quad \quad  \textbf{IF} $A_{i,0}$ has at least two distinct eigenvalues \\
\quad \quad \quad Split system $[A]$ as in Section~\ref{diagpfaff}; Update $\Phi$\\
 \quad \quad \quad  \textsc{Fmfs\_pfaff}  $(p, \tilde{A}^{11})$; Update $\Phi$, ${\{C_i\}}_{1 \leq i \leq n}$, ${\{Q_i\}}_{1 \leq i \leq n}$ \\
 \quad \quad \quad  \textsc{Fmfs\_pfaff}  $(p, \tilde{A}^{22})$; Update $\Phi$, ${\{C_i\}}_{1 \leq i \leq n}$, ${\{Q_i\}}_{1 \leq i \leq n}$\\
\quad \quad \textbf{ELSE IF} $A_{i,0}$ has one non-zero eigenvalue \\
 \quad \quad \quad  Update $Q_i$ from the eigenvalues of ${A}_{i,0}$\\
 \quad \quad \quad  $A(x) \gets $ Follow Section~\ref{shiftpfaff}  ($A_{i,0}$ is now nilpotent) \\
 \quad \quad \quad   \textsc{fmfs\_pfaff} $(p, \tilde{A}(x))$; Update $\Phi$, ${\{C_i\}}_{1 \leq i \leq n}$, ${\{Q_i\}}_{1 \leq i \leq n}$\\
\quad \quad \textbf{ELSE} \\
 \quad \quad \quad  Apply rank reduction of Section~\ref{sec:rankred}; Update
      $\Phi$; Update $p$; Update $A_{i,0}$ \\
 \quad \quad \quad  \textbf{IF} $p_{i} >0$ and $A_{i,0}$ has at least two distinct eigenvalues \\
 \quad \quad \quad  \quad   Split system as in Section~\ref{diagpfaff}; Update $\Phi$ \\
 \quad \quad \quad  \quad  \textsc{fmfs\_pfaff}  $(p, \tilde{A}^{11}(x))$; Update $\Phi$, ${\{C_i\}}_{1 \leq i \leq n}$, ${\{Q_i\}}_{1 \leq i \leq n}$ \\
 \quad \quad \quad  \quad  \textsc{fmfs\_pfaff}  $(p, \tilde{A}^{22}(x))$; Update $\Phi$, ${\{C_i\}}_{1 \leq i \leq n}$, ${\{Q_i\}}_{1 \leq i \leq n}$\\
\quad \quad \quad  \textbf{ELSE IF} $A_{i,0}$ has one non-zero eigenvalue \\
 \quad \quad \quad  \quad  Update $Q_i$ from the eigenvalues of $A_{i,0}$ \\
 \quad \quad \quad  \quad $A(x) \gets $ Follow Section~\ref{shiftpfaff}; ($A_{i,0}$ is now nilpotent) \\ \quad \quad \quad  \quad  \textsc{Fmfs\_pfaff} $(p, \tilde{A}(x))$; Update $\Phi$, ${\{C_i\}}_{1 \leq i \leq n}$, ${\{Q_i\}}_{1 \leq i \leq n}$ \\
 \quad \quad \quad  \textbf{ELSE} \\
 \quad \quad \quad  \quad  Follow Section~\ref{sec:invariants} to compute $\omega(A_i) = \frac{{\ell} }{m}$ \\
 \quad \quad \quad  \quad  $x_{i} \gets {x_{i}}^m$ \\
 \quad \quad \quad  \quad  Apply rank reduction of Section~\ref{sec:rankred}\\
 \quad \quad \quad  \quad   Update $\Phi$; Update $p$ ($p_{i} \gets {\ell} $); Update $A_{i,0}$ \\
 \quad \quad \quad  \quad  Update $Q_i$ from the eigenvalues of ${A}_{i,0}$\\
 \quad \quad \quad  \quad  $A(x) \gets $ Follow Section~\ref{shiftpfaff}; ($A_{i,0}$ is now nilpotent) \\
 \quad \quad \quad  \quad  \textsc{fmfs\_pfaff} $(p, A(x))$; Update $\Phi$, ${\{C_i\}}_{1 \leq i \leq n}$, ${\{Q_i\}}_{1 \leq i \leq n}$ \\
\quad \quad \quad  \textbf{END IF}\\
 \quad \quad \textbf{ END IF} \\
 \quad \textbf{ END WHILE} \\
\quad \textbf{ RETURN} {$p$, $A$, $\Phi$, ${\{C_i\}}_{1 \leq i \leq n}$, ${\{Q_i\}}_{1 \leq i \leq n}$}.
    \end{tabular} & \\
    \hline
  \end{tabular}
\end{center}
\egroup

\bgroup
\def\arraystretch{1.0}
\begin{center}
\ \\\vspace{3cm}
  \begin{tabular}{p{11cm}l}
    \hline\\[-2.5ex]
    \textbf{Algorithm 2: rankReduce(}$p,A$\textbf{)}\\
    \hline\\[-2.5ex]
    \begin{tabular}{lp{9cm}}
      \textbf{Input:} & $p_1\dots,p_n, A_1.\dots,A_n$ of~\eqref{eq:sys}.\\
      \textbf{Output:} & $T(x) \in GL_d({\rm R}_L)$ and an irreducible equivalent system
                         \{$T(A) $\} whose Poincar\'e rank is its true
                         Poincar\'e rank and the rank of its leading coefficient
                         matrices is minimal ($\mu(T(A)) = m(T(A))$)\\
    \end{tabular} & \\
    \hline\\[-2.5ex]
    \begin{tabular}{ll}
\quad  $T \leftarrow I_{d}$\\
\quad  \textbf{FOR} every $i$ from $1$ to $n$ \textbf{DO}\\
\quad \quad  $T_i \leftarrow I_{d}$; $p_i  \leftarrow$ Poincar\'e rank of $A_{i}$ \\
% \quad \quad \textbf{IF}  condition $(\mathcal{C})$ is satisfied \textbf{THEN}
% \\
 \quad \quad  $U(\bar{x_i})  \leftarrow $ yields the
      form~\eqref{gaussformPfaffian} \\
\quad \quad \textbf{IF} $U(\bar{x_i})$ cannot be determined \\
 \quad \quad\quad \textbf{RETURN
     ERROR} ``Column module not free.'' \textbf{END IF}\\
 \quad \quad   $A_{i} \gets  U^{-1} A_{i} U $; $T_i  \leftarrow  T_i U$ \\
 \quad \quad  \textbf{WHILE} $\det (G_{A_{i}} (\lambda))=0$ and $p_i >0$ \textbf{DO} \\
 \quad \quad \quad $Q(\bar{x_i}), \varrho  \leftarrow $
      Proposition~\ref{gauss3Pfaffian} \\
 \quad \quad \quad \textbf{IF} $Q(\bar{x_i})$ cannot be determined \\
 \quad \quad \quad\quad \textbf{RETURN
     ERROR} ``Row module not free.'' \textbf{END IF}\\
 \quad \quad \quad $S(x_i) \leftarrow $ Proposition~\ref{shearingPfaffian} \\
 \quad \quad \quad $P  \leftarrow  Q S$; $T_i  \leftarrow  T_i P$ \\
 \quad \quad \quad $A_{i} \leftarrow  P^{-1} A_{i} P - x_i^{p_i} S^{-1} \frac{\partial S}{\partial x_i}$ \\
 \quad \quad \quad $p_i  \leftarrow $ Poincar\'e rank of $A_{i}$ \\
 \quad \quad \quad  $U(\bar{x_i})  \leftarrow $ yields the form~\eqref{gaussformPfaffian}  \\
 \quad \quad \quad $A_{i} \leftarrow  U^{-1} A_{i} U $; $T_i  \leftarrow  T_i U$ \\
 \quad \quad \textbf{END WHILE} \\
% \quad \quad \textbf{END IF} \\
\quad \quad \textbf{FOR} every $j \neq i$ from $1$ to $n$ \textbf{DO}\\
\quad \quad \quad $A_{j} \leftarrow T_i^{-1} A_{j} T_i - x_j^{p_j+1} T_i^{-1} \pder{T_i}{x_j}$ \\
\quad \quad \textbf{END FOR}\\
\quad \quad $T \leftarrow T T_i$ \\
\quad \textbf{END FOR}\\
\quad \textbf{RETURN} ($T, p_1, \dots, p_n, A_1,\dots,A_n$).
    \end{tabular} & \\
    \hline
  \end{tabular}
\end{center}
\egroup

\subsection{An Alternative Rank Reduction Algorithm}
\label{sec:alt}
In the case of univariate systems, Levelt's investigations of the existence of
stationary sequences of free lattices lead to an algorithm which reduces the
Poincar\'e rank to its minimal integer value~\cite{key57}. This algorithm was
then generalized to the bivariate case by the first author of this paper et
al. in~\cite{key5}. The theoretical basis of this algorithm differs
substantially from the algorithm given herein based on Moser's criterion. The
final result of both approaches however, i.e.\ the algorithm itself, is based on
applying column reductions and shearing transformations in both algorithms,
though in a different manner. In fact, the algorithms coincide for the
particular case of $\varrho =0$. The limitation in~\cite{key5} within the
generalization to the multivariate case is in guaranteeing the freeness
conditions for the leading matrix coefficient $A_{i,0}$ as stated in
Section~\ref{proofpfaff}. The additional condition of the freeness of the row
module as in Proposition~\ref{gauss3Pfaffian} is not required. Since the linear algebra
problem is resolved in Section~\ref{colpfaff}, this results in Algorithm~$3$.

Although both algorithms have an identical cost~\cite[pp 108]{key73},
experimental results for the univariate case and certain bivariate systems
(singularly-perturbed linear differential systems) suggest that the
lattice-based algorithm complicates dramatically the coefficients of the system
under reduction, even if Moser's criterion is adjoined to avoid some unnecessary
computations~\cite[Section 4.3]{key102}. Hence, Algorithm~$2$ can be used as
long as the required freeness conditions hold. Nevertheless, if the freeness of
the row module of $G_{A_i}(\lambda=0)$ is not satisfied, then Algorithm~$3$ can
be used as long as the column modules of $A_{i,0}$ and $B^{11}$ is free.  There remains
however, the question on the equivalence of these conditions.

\bigskip
\bgroup
\def\arraystretch{1.0}
\begin{center}
  \begin{tabular}{p{11cm}l}
    \hline\\[-2.5ex]
    \textbf{Algorithm 3: rankReduce\_alt(}$p,A$\textbf{)}\\
    \hline\\[-2.5ex]
    \begin{tabular}{lp{9cm}}
    \textbf{Input:} & $p_1\dots,p_n, A_1.\dots,A_n$ of~\eqref{eq:sys}.\\
    \textbf{Output:} & $T(x) \in GL_d({\rm R}_L)$ and an irreducible equivalent system
                       \{$T(A) $\} whose Poincar\'e rank is its true
                       Poincar\'e rank and the rank of its leading coefficient
                       matrices is minimal ($\mu(T(A)) = m(T(A))$)\\
    \end{tabular} & \\
    \hline\\[-2.5ex]
    \begin{tabular}{ll}
\quad  $T \leftarrow I_{d}$\\
\quad  \textbf{FOR} every $i$ from $1$ to $n$ \textbf{DO}\\
\quad \quad  $T_i \leftarrow I_{d}$; $p_i  \leftarrow$ Poincar\'e rank of $A_{i}$ \\
\quad \quad \textbf{WHILE} $j < d-1$ and $p_i>0$ \textbf{DO} \\
% \quad \quad \quad \textbf{IF}  condition $(\mathcal{C})$ needed to arrive at the form~\eqref{gaussformPfaffian}  is satisfied \textbf{THEN}\\
\quad \quad \quad  $U(\bar{x_i})  \leftarrow $ yields the form
      \eqref{gaussformPfaffian} \\
\quad \quad \quad \textbf{IF} $U(\bar{x_i})$ cannot be determined \\
 \quad\quad \quad\quad \textbf{RETURN
     ERROR} ``Column module not free.'' \textbf{END IF}\\
 \quad \quad  \quad   $r = \operatorname{rank}(A_{i,0})$ \\
 \quad \quad  \quad   $S(x_i) \leftarrow $ Proposition~\ref{shearingPfaffian} with $\varrho =0$ (i.e.\ $S(x_i) \gets \operatorname{Diag}(x_i I_r, I_{d-r})$)\\
 \quad \quad  \quad   $P  \leftarrow  U S$\\
 \quad \quad \quad  $A_{i} \leftarrow  P^{-1} A_{i} P - x_i^{p_i} S^{-1} \frac{\partial S}{\partial x_i}$ \\
 \quad \quad \quad  $\tilde{p}_i  \leftarrow $ Poincar\'e rank of $A_{i}$ \\
 \quad \quad \quad   \textbf{IF}  $\tilde{p}_i  < p_i$ \textbf{THEN} \\
 \quad \quad \quad   \quad $j \leftarrow 0$ \\
 \quad \quad \quad   \textbf{ELSE}\\
 \quad \quad \quad   \quad $j \leftarrow j+1$ \\
 \quad \quad \quad  \textbf{END IF}\\
 \quad \quad \quad  $p_i \leftarrow \tilde{p}_i$ \\
\quad \quad \quad $T_i  \leftarrow  T_i P$ \\
\quad \quad \textbf{END WHILE} \\
\quad \quad \textbf{FOR} every $j \neq i$ from $1$ to $n$ \textbf{DO}\\
\quad \quad \quad $A_{j} \leftarrow T_i^{-1} A_{j} T_i - x_j^{p_j+1} T_i^{-1} \frac{\partial T_i}{\partial x_j}$ \\
\quad \quad \textbf{END FOR}\\
\quad \quad $T  \leftarrow   TT_i$ \\
\quad \textbf{END FOR}\\
\quad \textbf{RETURN} ($T, {A_{i}}_{\{1 \leq i \leq n\}}$).
    \end{tabular} & \\
    \hline
  \end{tabular}
\end{center}
\egroup

\goodbreak
\section{Conclusion}
\label{conpfaff}
In this article, we studied completely integrable Pfaffian systems with normal
crossings in several variables. We showed that one can associate a set of
univariate linear singular differential systems from which the formal invariants
of the former can be retrieved. This reduces computations to computations over a
univariate field via \textsc{Isolde} or \textsc{Lindalg}, and limits the numbers
of coefficients necessary for the computations. We then complemented our work
with a rank reduction algorithm based on generalizing Moser's criterion and the
algorithm given by Barkatou in~\cite{key41}. The former is applicable to any
bivariate system. However, for multivariate systems, it demands that several explicitly
described conditions are met.

One field of investigation is the possibility of weakening the conditions required in the multivariate setting for the rank reduction. Another one is the adaptation of the techniques developed herein for rank reduction to generalize the notion of \textit{simple} systems (see~\cite{key40} for
$m=1$). This notion, in the univariate case, gives another approach to construct a basis of the space of regular
solutions~\cite{key25}, and the obstacles encountered are similar to those in rank reduction. An additional field is to study closed form solutions~\cite{key3811}
in the light of associated ODS introduced in Section~\ref{sec:invariants}. 

In future work we aim
to thoroughly study the theoretical complexity of the proposed algorithm as well
as give detailed information of the effects on the truncation of the input
system during the computation. This work is not straightforward and even in the
case of regular systems, it is
not studied in the existing work (i.e. ~\cite[Chapter 3]{key73} and references therein). 

Systems arising from applications do not necessarily or directly fall into the
class of completely integrable Pfaffian systems with normal
crossings. Investigations in more general classes can be found
in~\cite{key32,key3062,key4092} and references therein. 

An additional field of investigation is the formal reduction in the difference
case using the approaches proposed herein. Praagman established
in~\cite{key3239} a formal decomposition of $m$ commuting partial linear
difference operators. This study was intended as an analog to that established
by Levelt, van den Essen, G\'erard, Charri\`ere, Deligne, and
others~\cite{key3,key53,key20,key21}. 

\section*{Acknowledgments}
\label{Ack}
We are grateful to the anonymous referees whose comments and suggestions
improved the readability of this paper. The second author would also like to
thank the Technische Universit\"at Wien, where he is employed by the time of
submission and supported by the ERC Starting Grant 2014 SYMCAR 639270. He
furthermore would like to thank Matteo Gallet for valuable discussions. The
third author would like to thank the University of Limoges since a part of this
work was done during the period of her doctoral studies at DMI.
%\pagebreak

\bibliographystyle{plain}

\end{document}